\documentclass[12pt]{article}
\usepackage[T1]{fontenc}

\usepackage{amsmath,amssymb,amsthm}
\usepackage{mathtools}
\usepackage{sgame}
\usepackage{enumitem}
\usepackage{tikz}
\usepackage{natbib}
\usepackage[colorlinks,citecolor=blue,urlcolor=blue]{hyperref}

\numberwithin{equation}{section}
\theoremstyle{plain}
\newtheorem{theorem}{Theorem}[section]
\newtheorem{lemma}{Lemma}[section]
\newtheorem{proposition}{Proposition}[section]
\newtheorem{remark}{Remark}[section]
\newtheorem{corollary}{Corollary}[section]
\newtheorem{definition}{Definition}[section]

\newcommand{\ICR}{{\rm R}} 
\newcommand{\ICRt}{{\rm ICR}} 
\newcommand{\BR}{{\rm br}} 
\newcommand{\AICR}{{\rm AR}} 

\newcommand\naturals{\mathbb N}

\newcommand{\marg}{\text{marg}} 
\newcommand{\id}{\text{id}} 

\begin{document}
\title{
Direct Representations for Interim Correlated Rationalizability\footnote{Olivier Gossner acknowledges support from the French National Research Agency (ANR), ``Investissements d'Avenir'' (ANR-11-IDEX-0003/LabEx Ecodec/ANR-11-LABX-0047) and from UKRI (SInfoNiA). The authors are grateful to Fran\c{c}oise Forges for stimulating discussions and comments.} \\[0.5cm] 
} 
\author{Olivier Gossner\thanks{CNRS - CREST - \'Ecole Polytechnique and London School of Economics} \quad Rafael Veiel\thanks{University of Texas at Austin}}
\date{\today}
\maketitle
\begin{abstract}
We study direct representations of information for interim correlated rationalizability. For a fixed finite payoff structure, each type induces a hierarchy of surviving action sets. Pushing the common prior through this map projects the information structure onto the solution concept's output language. When best-response regions are convex, this representation is direct: the solution concept applied to the hierarchy seen as a type is the identity. The induced distributions are characterized by level-by-level obedience constraints. Terminal \ICRt\ sets alone do not have this property. For arbitrary finite payoff structures, we refine each hierarchy level with a tag identifying a convex cell of its best-response region. Augmented hierarchies provide a direct representation and project onto the ordinary hierarchy. Full augmented hierarchies may form a continuum, but retaining only the tags at the boundaries of constant stretches of the ordinary hierarchy yields an exact countable representation with finitely many obedience constraints per type. Finite-type models are dense in terminal rationalizability outcome distributions.
\end{abstract}

 \section{Introduction}\label{sec:intro}

Interim correlated rationalizability, or \ICRt, captures how higher-order beliefs restrict behavior under common certainty of rationality \citep{dekel2007interim}. The higher-order-belief logic underlying \ICRt\ is central to global-games analysis \citep{carlsson1993global,morris2003global}, informational robustness \citep{bergemann2013robust,oyama2020generalized}, and the design of information and incentives to obtain outcomes robust to strategic uncertainty \citep{halac2021rank,halac2022addressing}.

A common question in these literatures is to characterize the induced behaviors across all common-prior information structures for a given payoff structure. For \ICRt, this entails searching over arbitrary type spaces and the higher-order beliefs they encode, a domain that is vast and generally intractable. In binary-action supermodular games, \citet{morris2020implementation} show that a related information-design problem can be reduced to state--action distributions satisfying obedience and sequential obedience. Outside such special environments, no general exact reduction of the information-structure search is available for the \ICRt\ calculation in arbitrary finite games.

Our benchmark is correlated equilibrium and its incomplete-information extensions~\citep{aumann1974subjectivity,aumann1987correlated,forges1986approach,forges1993five,bergemann2016bayes}. There, every induced state--action distribution can itself serve as an information structure: draw the state and action profile according to that distribution and privately reveal to each player her action. Following the recommendations reproduces the distribution, and every distribution that self-implements in this way is induced by itself viewed as an information structure. This construction provides two reductions: the outcome distribution itself is a direct information structure, and feasible distributions are characterized by linear obedience inequalities written directly on the state--action distribution. These properties are fundamental to comparing information structures and to information design \citep{bergemann2019information}. Our objective is to obtain both for \ICRt: a direct representation in payoff-dependent records and an obedience characterization of the distributions over those records.

To identify the relevant payoff-dependent records, recall that the behavioral implications of common certainty of rationality can be computed type by type by iterated deletion of never-best responses in the agent form, where each type is treated as a separate agent. Starting from all actions, at each round a type retains the actions that are optimal under some joint distribution consistent with her conditional belief about the state and the opponents' surviving actions. The resulting decreasing sequence $(\ICR_i^m(t_i))_{m\geq0}$ is the ordinary \ICRt\ hierarchy of type $t_i$, and its intersection $\ICR_i^\infty(t_i)=\bigcap_m\ICR_i^m(t_i)$ is her terminal \ICRt\ set. These are distinct candidate records: the terminal set gives the eventual rationalizable actions, whereas the hierarchy also records the finite-order reasoning that generates them.

Given a candidate record, we map each type to its record and push the common prior forward through this map. The representation is direct if, when the resulting distribution is itself viewed as a common prior on record types, recomputing the calculation reproduces the announced records; the induced record map is then the identity. Directness concerns the reproduction of records. A separate question is which distributions over records can arise from an underlying information structure. We characterize these distributions by conditional best-response requirements, which we call obedience constraints. Two further dimensions matter: the cardinality of the direct type space and the number of nontrivial constraints required for each type.

Our main result gives an exact countable representation for every finite payoff structure. We construct a countable direct refinement of ordinary \ICRt\ hierarchies. A distribution of the state and ordinary hierarchies is induced by a common-prior information structure if and only if there exists a distribution over the refined records satisfying obedience whose marginal on the state and ordinary hierarchies is the original distribution. At most $2|A_i|$ nontrivial obedience constraints are required for player $i$ and each type. The theorem therefore replaces the search over arbitrary information structures by a search over distributions on a countable payoff-dependent domain subject to finitely many conditions per type.

The terminal \ICRt\ set is the simplest candidate record, but it need not be direct. In a version of the electronic-mail game of \citet{rubinstein1989electronic}, an initial elimination propagates along a chain of types: each subsequent elimination is sustained by an earlier elimination for the preceding type. All types nevertheless have the same terminal set. Projecting types only onto that set removes the information supporting the chain of eliminations, and recomputing \ICRt\ enlarges it. The ordinary hierarchy is therefore the next candidate record.

Whether this record is direct depends on whether averaging the beliefs of types with the same hierarchy preserves their best-response sets. This leads us to distinguish $\BR$-convex payoff structures, for which ordinary \ICRt\ hierarchies are direct, from general finite payoff structures, where ordinary hierarchies may fail to be direct and must be refined.

A best-response region is the set of beliefs that generate a given action set; a payoff structure is $\BR$-convex when every such region is convex. Conditional on an ordinary hierarchy, the projected belief averages original beliefs that generated the same action set at every round. Under $\BR$-convexity, this average remains in the same region, so ordinary \ICRt\ hierarchies are direct. Every binary-action payoff structure is $\BR$-convex, as are ordered environments with interval best-response sets; concavity in own action together with increasing differences provides a primitive sufficient condition.

Having established directness, we turn to the size of the representation and the number of obedience constraints. The space of ordinary \ICRt\ hierarchies is countable independently of $\BR$-convexity. A decreasing sequence of subsets of a finite action set has finitely many strict decreases and is determined by a finite chain of action sets and their transition dates. In the $\BR$-convex case, the direct characterization initially imposes obedience at every level. Monotonicity of best responses makes the intermediate constraints on each constant stretch redundant: obedience at its entry and finite exit, together with a terminal condition, implies obedience throughout the stretch. A hierarchy with $q_i$ constant stretches therefore requires at most $2q_i\leq2|A_i|$ nontrivial conditions for player $i$.

For general finite payoff structures, ordinary hierarchies need not be direct: averaging beliefs within a nonconvex best-response region may leave that region. We restore convexity by partitioning every best-response region into finitely many convex cells and recording the cell containing the relevant belief at each round. The resulting full augmented hierarchy is direct for the augmented calculation and projects level by level onto the ordinary hierarchy. Full augmentation has two shortcomings. Its type space may have the cardinality of the continuum because cell tags can change indefinitely while the ordinary action set remains constant, and its level-by-level characterization may require countably many obedience constraints per type.

Both shortcomings come from recording a cell tag at every round, including within constant stretches of the ordinary hierarchy. The ordinary hierarchy already records the successive action sets and their transition dates, while monotonicity makes obedience at intermediate rounds redundant. A transition-tagged hierarchy therefore preserves the ordinary hierarchy but records a cell tag only at the entry and finite exit of each constant stretch. These tags retain the convex-cell information needed for directness while discarding all intermediate tags. Transition-tagged hierarchies form a countable direct type space. Obedience need only be imposed at the recorded transitions, together with one terminal condition, so at most $2|A_i|$ nontrivial conditions are required for each player and type. Projecting away the tags preserves the exact state--hierarchy distribution.

The closest information-design benchmark is \citet{morris2020implementation}. For binary-action supermodular games, they use obedience and sequential obedience on a countable domain to characterize the closure of state--action distributions induced by the smallest equilibrium. We instead provide an exact representation of state--\ICRt\ hierarchy distributions for arbitrary finite games. The represented objects differ: we do not characterize the projection onto rationalizable state--action distributions.

Countability cannot generally be replaced by finiteness for exact hierarchy representation. The electronic-mail example generates infinitely many ordinary hierarchies through unbounded transition dates, so no finite type space can reproduce its state--hierarchy distribution exactly. This obstruction concerns exact representation of the entire hierarchy. At the level of terminal behavior, for every fixed payoff structure, finite common-prior models are dense, in total variation, in the induced distributions of the state and terminal \ICRt\ sets.

The approximation result complements \citet[Theorem~3]{dekel2006topologies}. They show that finite types are dense without the common-prior restriction in the strategic topology, which compares behavior across games, while finite common-prior types are not dense in that topology. We fix one payoff structure, retain the common-prior restriction, and approximate distributions of the state and terminal \ICRt\ sets rather than types across games.

The paper is organized as follows. Section~\ref{sec:rel_lit} discusses related literature. Section~\ref{sec:model} defines the model and gives the two motivating examples. Section~\ref{sec:brconvex} establishes direct representation by ordinary hierarchies under \BR-convexity, reduces obedience to transition rounds, and identifies polyhedral classes. Section~\ref{sec:generalRP} constructs full augmented hierarchies and then obtains the general exact countable representation with transition-tagged hierarchies. Section~\ref{sec:finite-approximation} establishes finite approximation.

\section{Related Literature}\label{sec:rel_lit}

The direct-representation problem studied here connects four strands of work: canonical representations of information, correlated equilibrium and information design, rationalizability, and approximation of belief hierarchies.

Harsanyi type spaces encode beliefs about payoff-relevant uncertainty and other players' types \citep{harsanyi1967games,harsanyi1968games,Harsanyi1968GamesWI}, while the universal type space of \citet{mertens1985formulation} provides a canonical space of coherent belief hierarchies. That space is game-independent and, in nondegenerate environments, has the cardinality of the continuum. Our final direct space depends on the fixed payoff structure and is countable. \citet{dekel2007interim} show that canonical belief hierarchies determine \ICRt\ behavior across games, and \citet{ely2006hierarchies} construct a universal type space for interim rationalizability under a different correlation restriction. \citet{govei2026strategic} axiomatize game-dependent strategic type spaces and characterize feasible ordinary \ICRt\ hierarchies through a finite automaton. The present paper asks when such payoff-dependent records reproduce themselves under a common prior and how they must be refined when ordinary hierarchies do not. Related approaches compress information by its decision or equilibrium consequences: posterior beliefs represent signals in single-agent problems \citep{blackwell1953equivalent}, while \citet{gossner2000comparison} and \citet{GosMer01} compare information structures through their consequences across games.

In correlated equilibrium, a recommendation can serve as the signal and obedience characterizes the resulting distributions \citep{aumann1974subjectivity,aumann1987correlated}; \citet{BrandenburgerDekel1987} connect correlated rationalizability to subjective correlated equilibrium. Communication equilibrium and Bayes correlated equilibrium extend direct recommendation and obedience to incomplete-information environments \citep{forges1986approach,forges1993five,bergemann2016bayes,bergemann2019information}. \citet{mathevet2020information} formulate information design through distributions over belief hierarchies. Their formulation covers environments in which action recommendations do not suffice. We instead project information into the payoff-dependent output language of \ICRt\ and require that the projected calculation reproduce itself. The broader revelation-principle background comes from mechanism design \citep{Myerson1979,Myerson1982}; \citet{hammond1994revelation} studies implementation under Bayesian rationalizable strategies.

Rationalizability was introduced by \citet{bernheim1984rationalizable} and \citet{pearce1984rationalizable}; \citet{dekel2007interim} define its interim correlated version. Neighboring work develops epistemic foundations and belief restrictions \citep{BattigalliSiniscalchi2003,battigalli2011interactive}, belief-free rationalizability and informational robustness \citep{bergemann2017belief}, interim partially correlated rationalizability \citep{Tang2015}, infinite games \citep{WeinsteinYildiz2017}, informational robustness under common belief in rationality \citep{ziegler2022informational}, and large games \citep{balbus2025interimcorrelatedrationalizabilitylarge}. \citet{halac2021rank} use rank uncertainty to implement work uniquely in a team-production problem, while \citet{halac2022addressing} formulate the design of incentives and information as inducing a desired action profile as the unique rationalizable outcome. We study the geometry of the best-response operator and use finite convex refinements when best-response regions are nonconvex. The closest operational benchmark is \citet{morris2020implementation}. For binary-action supermodular games, they characterize the closure of smallest-equilibrium implementable state--action distributions through obedience and sequential obedience and construct countable implementing information structures. We treat arbitrary finite games and represent exact state--\ICRt\ hierarchy distributions on a countable transition-tagged domain with finitely many obedience constraints per type. We do not characterize the projection onto rationalizable state--action distributions.

\citet{lipman2003finite} shows that, in finite partitions models, any fixed finite order of beliefs and knowledge consistent with common support can be matched by another finite model with a common prior. The construction changes the higher-order tail, and product-topology proximity need not preserve strategic behavior. Our truncations are instead payoff-dependent: they retain the terminal \ICRt\ set together with a finite prefix of augmented tags and approximate the induced terminal distribution in total variation. \citet[Theorem~3]{dekel2006topologies} show that finite types are dense in the strategic topology, which compares \ICRt\ behavior across bounded finite games, while finite common-prior types are not. Their positive result implies finite approximation for each fixed game without the common-prior restriction; our result supplies the fixed-game, common-prior counterpart for terminal \ICRt\ distributions. \citet{bergemann2024strategic} establish a different common-prior approximation result: finite simple information structures are dense in the almost-common-knowledge topology that makes equilibrium outcomes continuous. Their object is equilibrium behavior across nearby information structures, whereas ours is the terminal \ICRt\ distribution of a fixed payoff structure.

\section{Model}\label{sec:model}

This section defines payoff structures, common-prior information structures, and the \ICRt\ hierarchy. It then presents the two obstructions that motivate our direct representations and records the geometric properties of the best-response operator used below.

\paragraph{Notation}
 For a family of sets $(X_i)_{i}$, we let $X= \prod_i X_i$ and $X_{-i} = \prod_{j\neq i} X_j$, for every $i$. Given a measurable set $X$, $\Delta(X)$ denotes the set of probability distributions on $X$. The set of nonnegative integers is denoted $\naturals$. 

For a family of maps $(f_i\colon X_i\to Y_i)_i$, we write $f=(f_i)_i$ for their product and $f_{-i}=(f_j)_{j\neq i}$ for the product excluding player $i$.

A marginal on coordinates $x_1,\dots, x_n$ of a distribution $P\in \Delta(\prod_\ell X_\ell)$ is denoted $\marg_{x_1,\dots,x_n}(P)$.

For a measurable map $f\colon X\to Y$ and a probability measure $\mu\in\Delta(X)$, we denote by $f_\#\mu\in\Delta(Y)$ the pushforward measure, defined by $(f_\#\mu)(E)=\mu(f^{-1}(E))$ for every measurable $E\subseteq Y$. We also write $f_\#$ for the induced map from $\Delta(X)$ to $\Delta(Y)$.

\paragraph{Payoff structures}

We fix a finite set $N$ of players and a finite set $K$ of states of nature. A \emph{payoff structure} $u$ consists of a finite action set $A_i$ and a payoff function $u_i\colon K\times A\to \mathbb{R}$ for each player $i$. For every player $i$, let $\mathcal{A}_i \coloneqq 2^{A_i} \setminus \{ \varnothing\}$ denote the set of nonempty action sets.

\paragraph{Common-prior information structures}

A \emph{common-prior information structure} consists of a family of standard Borel type spaces $(T_i)_{i\in N}$ and a probability distribution $P\in\Delta(K\times T)$. We call it a \emph{common prior} when the type spaces are understood. For every player $i$, fix a measurable regular conditional probability kernel
$$
P_i\colon T_i\to\Delta(K\times T_{-i})
$$
induced by $P$. We also write $P(\cdot\mid t_i)=P_i(t_i)$. The selected kernel is defined on all of $T_i$, so the recursive constructions below are defined pointwise. Its interpretation as a conditional distribution induced by $P$, and all identities based on disintegration, are understood outside a set of zero probability under the marginal of $P$ on $T_i$.

\paragraph{Interim Correlated Rationalizability}

Fix a payoff structure $u$ and a common prior $P$ with type spaces $(T_i)_i$. Interim Correlated Rationalizability (ICR) \citep{dekel2007interim} iteratively eliminates never best responses for every type of every player.\footnote{Our presentation slightly differs from \citep{dekel2007interim} but the two definitions are equivalent.} For a belief $p\in\Delta(K\times\mathcal A_{-i})$, let $\mathcal Q_i(p)$ be the set of \emph{compatible joint distributions} $q\in\Delta(K\times\mathcal A_{-i}\times A_{-i})$ satisfying
$$
\marg_{K,\mathcal A_{-i}}q=p
\quad\text{and}\quad
q\bigl(\{(k,B_{-i},a_{-i}):a_{-i}\in B_{-i}\}\bigr)=1.
$$
To define the \ICRt\ iteration, we use the following best-response operator on beliefs over opponent action sets. Player $i$'s best-response map $\BR_i\colon \Delta(K \times \mathcal{A}_{-i}) \to  \mathcal{A}_i$ is defined by
\begin{equation}\label{BR}
\BR_i(p) \coloneqq \bigcup_{q\in\mathcal Q_i(p)}
\arg\max_{a_i\in A_i}\mathbb E_q\bigl[u_i(k,a_i,a_{-i})\bigr].
\end{equation}
For every $B_i\in\mathcal A_i$, its \emph{best-response inclusion region} is
$$
\mathcal I_i(B_i)
\coloneqq
\{p\in\Delta(K\times\mathcal A_{-i}):B_i\subseteq\BR_i(p)\},
$$
and its \emph{best-response region} is $\BR_i^{-1}(B_i)$. These regions satisfy
$$
\BR_i^{-1}(B_i)
=
\mathcal I_i(B_i)
\setminus
\bigcup_{a_i\in A_i\setminus B_i}
\mathcal I_i(B_i\cup\{a_i\}).
$$

The \emph{\ICRt\,hierarchy} $(\ICR^m_i(t_i))_{m \geq 0}$ of a type $t_i\in T_i$ is defined iteratively:
\begin{enumerate}
    \item  [1.] For every  $i\in N$ and $t_i\in T_i$, $\ICR^0_i(t_i) \coloneqq A_i$,
    \item  [2.] For every $m\in \naturals$, the conditional probability $P_i(t_i)\in \Delta(K\times T_{-i})$ and the measurable map $\ICR_{-i}^{m}$ induce the belief $(\id \times \ICR^{m}_{-i})_\#P_i(t_i)$ on $K\times \mathcal{A}_{-i}$. Given this belief,
    \begin{equation}\label{ICR}
    \ICR^{m+1}_i(t_i)
    \coloneqq
    \BR_i\left((\id \times \ICR^{m}_{-i})_\#P_i(t_i)\right).
    \end{equation}
\end{enumerate}

Write
$$
\ICR_i(t_i)\coloneqq\bigl(\ICR_i^m(t_i)\bigr)_{m\in\naturals},
\qquad
\ICR\coloneqq(\ICR_i)_i,
$$
for the full hierarchy maps.

The \emph{terminal \ICRt\ set} of type $t_i$ is
\begin{equation}
\ICR_i^\infty(t_i)
\coloneqq
\bigcap_{m\in\naturals}\ICR_i^m(t_i).
\end{equation}

\begin{lemma}[Properties of the ICR iteration]\label{lem:ICR-properties}
For every common prior $P$, every player $i$, and every $m\in\naturals$, the map
$$
\ICR_i^m\colon T_i\to\mathcal A_i
$$
is measurable. Moreover, for every $t_i\in T_i$,
$$
\varnothing\neq\ICR_i^{m+1}(t_i)\subseteq\ICR_i^m(t_i).
$$
Consequently, $\ICR_i^\infty(t_i)$ is nonempty for every $t_i$, and the map
$$
\ICR_i^\infty\colon T_i\to\mathcal A_i
$$
is measurable.
\end{lemma}

\begin{proof}
Nonemptiness follows because every $B_{-i}\in\mathcal A_{-i}$ is nonempty, so $\mathcal Q_i(p)$ is nonempty for every $p$, and the finite set $A_i$ contains a maximizer.

We prove monotonicity by induction. The claim is immediate at level zero. Suppose that $\ICR_j^m(t_j)\subseteq\ICR_j^{m-1}(t_j)$ for every player $j$ and every $t_j\in T_j$. Fix a player $i$, a type $t_i\in T_i$, and $a_i\in\ICR_i^{m+1}(t_i)$. Let $q^m$ be a compatible joint distribution under which $a_i$ is optimal, and consider the distribution of
$$
(k,\ICR_{-i}^{m-1}(t_{-i}),\ICR_{-i}^m(t_{-i}))
$$
under $P_i(t_i)$. Couple this distribution and $q^m$ conditionally on $(k,\ICR_{-i}^m(t_{-i}))$, and retain the marginal on $(k,\ICR_{-i}^{m-1}(t_{-i}),a_{-i})$. Since $\ICR_{-i}^m(t_{-i})\subseteq\ICR_{-i}^{m-1}(t_{-i})$, the resulting joint distribution is compatible with the level-$(m-1)$ belief and induces the same marginal on $K\times A_{-i}$ as $q^m$. Hence $a_i\in\ICR_i^m(t_i)$.

Measurability of $\ICR_i^m$ follows from \citet[Lemma~1]{dekel2007interim}, stated for the equivalent forecast formulation of \ICRt. Since $A_i$ is finite and
$$
\{t_i:a_i\in\ICR_i^\infty(t_i)\}
=
\bigcap_{m\geq0}
\{t_i:a_i\in\ICR_i^m(t_i)\},
$$
the terminal correspondence is measurable as well. Finally, monotonicity and finiteness of $A_i$ imply that every hierarchy stabilizes, so $\ICR_i^\infty(t_i)$ is nonempty.
\end{proof}

\subsection{Examples}

\paragraph{Example 1: The Electronic-Mail Game.} We first illustrate why terminal \ICRt\ sets do not provide a direct representation. Consider a version of the email game in \citet{rubinstein1989electronic}. There are two states $k_G,k_B$, two players, and two actions $\{a,b\}$ for each player. Payoffs are
\begin{center}\hfill
\begin{game}{2}{2}[$k_G$]
 &  $a$ & $b$ \\
$a$ & $ 1,1$ & $0,0$ \\
$b$ & $0,0$ & $1,1$
\end{game}\hfill~
\end{center}
in state $k_G$ and
\begin{center}\hfill
\begin{game}{2}{2}[$k_B$]
 &  $a$ & $b$ \\
$a$ & -1,-1 & $0,0$ \\
$b$ & $0,0$ & $1,1$
\end{game}\hfill~%
\end{center}%
in state $k_B$. Let $T_1=\{0,2,4,\ldots\}$ and $T_2=\{1,3,5,\ldots\}$. For $\varepsilon<1/2$, define
$$
P(k,t_1,t_2)
=
\begin{cases}
\varepsilon
&\text{if }(k,t_1,t_2)=(k_B,0,1),\\
(1-\varepsilon)^{\min\{t_1,t_2\}}\varepsilon
&\text{if }k=k_G,\ t_1>0,\ |t_1-t_2|=1,\\
0&\text{otherwise.}
\end{cases}
$$
Type $0$ of player 1 knows that the state is $k_B$, so $\ICR_1^1(0)=\{b\}$. Conditional on each subsequent type, the probability of the preceding type in the chain is $1/(2-\varepsilon)>1/2$. Induction along the chain therefore gives $\ICR_i^\infty(t_i)=\{b\}$ for every player and type. The same induction shows that the round at which action $a$ is first eliminated is unbounded along the chain, so the information structure induces infinitely many distinct hierarchies.

Revealing only the terminal set consequently gives every player the same signal $\{b\}$. The signal contains no information, and the probability of $k_B$ is $\varepsilon<1/2$. Action $a$ is then a best response under the compatible distribution assigning probability one to the opponent's action $a$, while $b$ is a best response under the one assigning probability one to $b$. Hence $\ICR_i^1(\{b\})=\{a,b\}$, and the iteration remains at $\{a,b\}$. The terminal-set distribution does not implement itself.

\paragraph{Example 2: Failure of Convex Best-Response Regions.}
We now show that best-response regions need not be convex. Consider a two-player payoff structure with one state $K=\{k\}$. Player 1 has actions $\{a,b,c\}$ and player 2 has actions $\{x,y,z\}$. Player 2's payoffs are irrelevant, and player 1's payoffs are reported in the left panel of Figure~\ref{fig:nonconvex-br}.
Let $p$ be the belief on $K \times \mathcal{A}_2$ that assigns probability one to $(k,\{x,y\})$, let $p'$ assign probability one to $(k,\{x,z\})$, and let $r$ assign probability one to $(k,\{y,z\})$.

\begin{figure}[htbp]
\centering
\begin{minipage}[c]{0.31\textwidth}
\centering
\[
\begin{array}{c|ccc}
 & x & y & z \\
\hline
a & 1 & 1 & 0 \\
b & 1 & 0 & 1 \\
c & \frac{2}{3} & \frac{2}{3} & \frac{2}{3}
\end{array}
\]
\end{minipage}
\hfill
\begin{minipage}[c]{0.65\textwidth}
\centering
\begin{tikzpicture}[x=.92cm,y=.82cm]
  \coordinate (P) at (0,0);
  \coordinate (Pp) at (7,0);
  \coordinate (R) at (3.5,4.2);

  \fill[brown!16] (P)--(Pp)--(R)--cycle;
  \draw[brown!75!black,very thick] (P)--(Pp)--(R)--cycle;

  \fill[red!16]
    (3.5,4.2)--(5.83,1.4)--(4.67,0)--(2.33,0)--(1.17,1.4)--cycle;
  \draw[red!65!black,very thick]
    (3.5,4.2)--(5.83,1.4)--(4.67,0)--(2.33,0)--(1.17,1.4)--cycle;

  \fill[blue!70!black] (P) circle (2.7pt)
    node[below left,font=\scriptsize] {$p$};
  \fill[blue!70!black] (Pp) circle (2.7pt)
    node[below right,font=\scriptsize] {$p'$};
  \fill[blue!70!black] (R) circle (2.7pt)
    node[above,font=\scriptsize] {$r$};
  \fill[red!70!black] (3.5,0) circle (2.7pt)
    node[above right,font=\scriptsize] {$\bar p$};
\end{tikzpicture}
\end{minipage}
\caption{Failure of convexity of a best-response region. The left panel reports player 1's payoffs. The right panel depicts the triangle spanned by $p,p'$, and $r$. In the brown zones, $\BR_1(q)=\{a,b\}$; in the red zone, $\BR_1(q)=\{a,b,c\}$.}
\label{fig:nonconvex-br}
\end{figure}
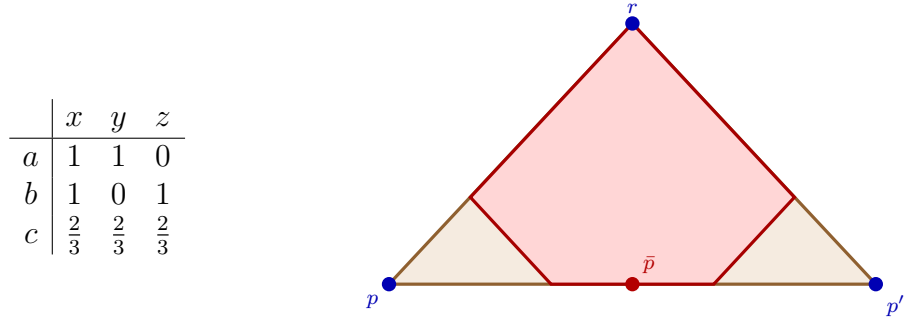

Under $p$, the compatible distributions over player 2's actions are exactly the distributions supported on $\{x,y\}$. Against any such distribution, action $a$ gives payoff $1$, action $b$ gives payoff at most $1$, and the safe action $c$ gives payoff $2/3$. Hence $a$ is always a best response. Moreover, the compatible distribution assigning probability one to $x$ makes both $a$ and $b$ best responses. Therefore
$\BR_1(p)=\{a,b\}$. Similarly, we have $\BR_1(p')=\{a,b\}$.

Now let
$\bar p=\frac12p+\frac12p'.
$ Under $\bar p$, there is a compatible joint distribution that assigns player 2's action $y$ after the event $\{x,y\}$ and $z$ after the event $\{x,z\}$. Its marginal on player 2's actions assigns probability $1/2$ to $y$ and probability $1/2$ to $z$. The expected payoffs of $a,b,c$ are then $\frac12,\frac12, \frac23$, respectively.

Thus $c$ is a strict best response under a compatible joint distribution for $\bar p$, so $c\in\BR_1(\bar p)$. Consequently,
$$
\BR_1(p)=\BR_1(p')=\{a,b\},
\qquad
\BR_1\!\left(\frac12p+\frac12p'\right)\neq \{a,b\}.
$$
Hence the best-response region $\BR_1^{-1}(\{a,b\})$ is not convex.

\paragraph{Best-response geometry.}
We use the following terminology in Euclidean spaces. A \emph{polyhedron} is a set defined by finitely many weak affine inequalities, and a \emph{mixed polyhedron} is a set defined by finitely many weak and strict affine inequalities. We allow inconsistent systems, so the empty set is both a polyhedron and a mixed polyhedron.

Following \citet{govei2026strategic}, we record three geometric properties of the best-response operator. The closed best-response inclusion regions support the limiting and pooling arguments below, while the representation of best-response regions as finite unions of mixed polyhedra supports the tagged construction. The proof is in Appendix~\ref{app:br-geometry}.

\begin{lemma}\label{Lem_BRpolytope}
For every player $i$, every action $a_i\in A_i$, and every $B_i\in\mathcal A_i$:
\begin{enumerate}[label=(\roman*)]
\item $\mathcal I_i(\{a_i\})$ is a polyhedron;
\item $\mathcal I_i(B_i)$ is a polyhedron;
\item $\BR_i^{-1}(B_i)$ is a finite union of mixed polyhedra.
\end{enumerate}
\end{lemma}

\section{Direct Representation under \texorpdfstring{$\BR$}{BR}-Convexity}\label{sec:brconvex}

The first example shows that terminal \ICRt\ sets discard information needed to reproduce \ICRt, while the second identifies the obstruction to using full hierarchies as direct types: averaging beliefs within a best-response region may leave that region. This section gives a complete answer when best-response regions are convex. Ordinary \ICRt\ hierarchies are direct types and form a countable space. After establishing the direct representation, we use monotonicity of best responses to reduce obedience to the entries and exits of constant stretches of the hierarchy. We then identify binary-action and ordered payoff structures for which the relevant best-response regions admit explicit linear descriptions.

\begin{definition}[$\BR$-convex and polyhedral payoff structures]
The payoff structure is \emph{$\BR$-convex} if $\BR_i^{-1}(B_i)$ is convex for every player $i$ and every $B_i\in\mathcal A_i$. It is \emph{polyhedral} if $\BR_i^{-1}(B_i)$ is a mixed polyhedron for every player $i$ and every $B_i\in\mathcal A_i$.
\end{definition}

\begin{remark}
Every polyhedral payoff structure is $\BR$-convex.
\end{remark}

Every common prior $P$ induces, by pushing forward under $\id\times\ICR$, the distribution
$$
P_{\ICR}
\coloneqq
(\id\times\ICR)_\#P
$$
on $K\times\mathcal A^{\naturals}$. Let $\mathcal P$ denote the set of distributions obtained in this way as the common prior and the type spaces vary.

\begin{definition}[\ICRt\,hierarchies]
The set of positive-probability \ICRt\ hierarchy profiles is
$$
S
\coloneqq
\left\{
s\in\mathcal A^{\naturals}:
P(K\times\{s\})>0
\text{ for some }P\in\mathcal P
\right\}.
$$
For every player $i$, let $S_i$ be the projection of $S$ on $\mathcal A_i^{\naturals}$.
\end{definition}

The space $S$ is countable. Indeed, every player's \ICRt\,hierarchy is a decreasing sequence of nonempty subsets of the finite set $A_i$. It therefore has finitely many strict decreases and is determined by a finite decreasing chain of action sets together with the finite sequence of rounds at which the strict decreases occur. There are finitely many such chains and countably many such sequences of rounds.

Every induced distribution $P_\ICR\in\mathcal P$ can itself be viewed as a common prior: $(k,s)$ is drawn according to $P_\ICR$, and player $i$ observes the hierarchy $s_i$. The representation is direct when recomputing \ICRt\ under this common prior reproduces the announced hierarchies, so $\ICR(s)=s$ almost surely. The following theorem characterizes all induced state--hierarchy distributions.

\begin{theorem}[Direct Representation by \ICRt\ Hierarchies]\label{RP1}
For any $\BR$-convex payoff structure and any $P \in \Delta(K\times \mathcal{A}^{\naturals})$, the following conditions are equivalent:
\begin{enumerate}
    \item $P \in \mathcal{P}$;
    \item When viewed as a common prior on hierarchy types, $P$ satisfies $\ICR(s)=s$ almost surely;
    \item $P(s^0{=}A) = 1$ and $P$ satisfies the family of \emph{obedience constraints}:
\begin{equation}\label{OC}
    s_i^m
    =
    \BR_i\left(
\marg_{k,s_{-i}^{m-1}}P(\cdot,\cdot\mid s_i)
    \right)
    \quad
    \text{for every $m\geq1$ and almost every $s_i$.}
\end{equation}
\end{enumerate}
\end{theorem}

\begin{proof}

(1. $\implies$ 3) Since $P\in\mathcal P$, there is a common prior $Q\in\Delta(K\times T)$ that induces $P$. Denote player $i$'s regular conditional belief under $Q$ by $Q_i(t_i)$. Since $\ICR_i^0(t_i)=A_i$ for every player $i$ and every type $t_i$, the induced distribution $P$ satisfies $P(s^0=A)=1$. Moreover, for every player $i$, every $t_i\in T_i$, and every $m\geq1$, the definition of the ICR iteration gives
\begin{equation}
\ICR_i^m(t_i)
=
\BR_i\left(
(\id\times\ICR_{-i}^{m-1})_\#Q_i(t_i)
\right).
\end{equation}
For every $B_i\in\mathcal A_i$, the inverse image $\BR_i^{-1}(B_i)\subseteq\Delta(K\times\mathcal A_{-i})$ is convex because the payoff structure is $\BR$-convex. Fix a player $i$, a round $m\geq1$, and a hierarchy $s_i$ of positive marginal probability. For almost every $t_i$ conditional on $\ICR_i(t_i)=s_i$, the belief
$$
q_i^m(t_i)
\coloneqq
(\id\times\ICR_{-i}^{m-1})_\#Q_i(t_i)
$$
belongs to $\BR_i^{-1}(s_i^m)$. Moreover, the law of iterated expectations gives
\begin{equation}
\marg_{k,s_{-i}^{m-1}}P(\cdot\mid s_i)
=
\int q_i^m(t_i)\,Q(dt_i\mid\ICR_i=s_i).
\end{equation}
The right-hand side is an average of beliefs in the convex set $\BR_i^{-1}(s_i^m)$ and therefore also belongs to this set. Hence
$$
s_i^m
=
\BR_i\left(
\marg_{k,s_{-i}^{m-1}}P(\cdot\mid s_i)
\right).
$$
Thus $P$ satisfies \eqref{OC}.

(3. $\implies$ 2) Regard $P$ as a common prior on $K\times\mathcal A^{\naturals}$ under which player $i$ observes the announced hierarchy $s_i$. The level-zero condition implies that the ICR hierarchy generated by $s_i$ starts at $s_i^0=A_i$. Suppose that its coordinates through round $m-1$ coincide with $(s_i^0,\ldots,s_i^{m-1})$. By \eqref{ICR} and \eqref{OC}, its coordinate at round $m$ is $s_i^m$. Induction yields $\ICR(s)=s$ $P$-almost surely.

(2. $\implies$ 1) Regard $P$ as the same common prior on hierarchy types. If $\ICR(s)=s$ almost surely, then the distribution induced by this common prior over states and ICR hierarchies is $P$ itself. Hence $P\in\mathcal P$.
\end{proof}

Condition \emph{2} is the typewise identity property defining directness. Condition \emph{3} expresses it through level-by-level obedience constraints. The equivalence of these conditions with being induced has the familiar revelation-principle interpretation.

\subsection{Obedience at Transitions}
\label{sec:number-oc}

The level-by-level characterization in Theorem~\ref{RP1} appears to require infinitely many obedience constraints for every hierarchy. Most of these constraints are redundant. Let
$$
\mathcal R_i(s_i)\coloneqq\{s_i^m:m\geq1\}
$$
be the collection of action sets reached at positive rounds. For every $B_i\in\mathcal R_i(s_i)$, the \emph{$B_i$-plateau} is the maximal interval of positive rounds on which $s_i^m=B_i$. Define its entry and exit rounds by
$$
\underline m_i(B_i;s_i)
\coloneqq
\min\{m\geq1:s_i^m=B_i\},
\qquad
\overline m_i(B_i;s_i)
\coloneqq
\sup\{m\geq1:s_i^m=B_i\}
\in\naturals\cup\{\infty\}.
$$
Thus $\overline m_i(B_i;s_i)=\infty$ precisely for the terminal \ICRt\ set
$$
s_i^\infty\coloneqq\bigcap_{m\geq0}s_i^m.
$$
For $P\in\Delta(K\times S)$, define
$$
p_i^m(s_i)
\coloneqq
\marg_{k,s_{-i}^{m-1}}P(\cdot,\cdot\mid s_i),
\qquad m\geq1,
$$
and
$$
p_i^\infty(s_i)
\coloneqq
\marg_{k,s_{-i}^\infty}
P(\cdot,\cdot\mid s_i).
$$

\begin{lemma}[Obedience at transitions]
\label{lem:transition-obedience-ordinary}
Fix $P\in\Delta(K\times S)$, a player $i$, and a hierarchy $s_i$ of positive marginal probability. Suppose that, for every $B_i\in\mathcal R_i(s_i)$,
$$
\BR_i\left(
p_i^{\underline m_i(B_i;s_i)}(s_i)
\right)
=B_i,
$$
that, for every $B_i\in\mathcal R_i(s_i)$ such that $\overline m_i(B_i;s_i)<\infty$,
$$
\BR_i\left(
p_i^{\overline m_i(B_i;s_i)}(s_i)
\right)
=B_i,
$$
and that
$$
s_i^\infty
\subseteq
\BR_i\bigl(p_i^\infty(s_i)\bigr).
$$
Then the level-by-level obedience constraints hold at $s_i$.
\end{lemma}

\begin{proof}
Fix $B_i\in\mathcal R_i(s_i)$, write $\ell=\underline m_i(B_i;s_i)$ and $r=\overline m_i(B_i;s_i)$, and let $\ell\leq m\leq r$. Under $P(\cdot\mid s_i)$, the opponent action sets defining $p_i^\ell(s_i)$, $p_i^m(s_i)$, and $p_i^r(s_i)$ are coupled by the same underlying hierarchy and are almost surely nested, where $p_i^r=p_i^\infty$ when $r=\infty$. Hence \citet[Lemma~4.1]{govei2026strategic} gives
$$
\BR_i(p_i^\ell(s_i))
\supseteq
\BR_i(p_i^m(s_i))
\supseteq
\BR_i(p_i^r(s_i)).
$$
If $r<\infty$, both outer terms equal $B_i$. If $r=\infty$, the first equals $B_i$ and the terminal condition implies that the last contains $B_i$. All inclusions are therefore equalities.
\end{proof}

\begin{corollary}[Transition obedience for ordinary hierarchies]
\label{cor:ordinary-transition-oc}
Suppose that the payoff structure is $\BR$-convex. A distribution $P\in\Delta(K\times S)$ belongs to $\mathcal P$ if and only if it satisfies the level-zero condition and, for every player and almost every hierarchy, the entry, finite-exit, and terminal conditions of Lemma~\ref{lem:transition-obedience-ordinary}.
\end{corollary}

\begin{proof}
The endpoint conditions imply all level-by-level obedience constraints by Lemma~\ref{lem:transition-obedience-ordinary}, so sufficiency follows from Theorem~\ref{RP1}. Conversely, level-by-level obedience implies the entry and finite-exit conditions. Moreover, $p_i^m(s_i)\to p_i^\infty(s_i)$ in total variation. Every action in $s_i^\infty$ is a best response at every finite round and remains a best response at the limit by Lemma~\ref{Lem_BRpolytope}(i).
\end{proof}

Every hierarchy contains at most $|A_i|$ plateaus. If it contains $q_i$ plateaus, the transition characterization uses $q_i$ entry conditions, $q_i-1$ finite-exit conditions, and one terminal condition. Thus at most $2q_i\leq2|A_i|$ nontrivial obedience conditions are required for each ordinary hierarchy, in addition to the level-zero condition.

By Lemma~\ref{Lem_BRpolytope}, each entry and finite-exit condition is equivalent to a finite disjunction of systems of weak and strict linear inequalities. The terminal condition is membership in the best-response inclusion region $\mathcal I_i(s_i^\infty)$ and is therefore a finite system of weak linear inequalities. For polyhedral payoff structures, each entry and finite-exit condition is a single system of weak and strict linear inequalities.

We next identify two natural classes of polyhedral payoff structures. With binary actions, best-response regions are described by two linear payoff-difference functionals. With ordered actions, interval best-response sets yield mixed-polyhedral regions, and concavity in own action together with increasing differences provides sufficient conditions for the interval property.

\subsection{Binary-Action Payoff Structures}\label{sec:binary}

A finite payoff structure is \emph{binary-action} if every player has exactly two actions. The finite state space $K$ is otherwise unrestricted.
\begin{proposition}[Binary-action payoff structures]
Every binary-action payoff structure is polyhedral, and thus $\BR$-convex.
\end{proposition}

\begin{proof}
Fix a player $i$ and write $A_i=\{a_i,b_i\}$. For every $(k,a_{-i})\in K\times A_{-i}$, let
$$
d_i(k,a_{-i})
\coloneqq
u_i(k,a_i,a_{-i})-u_i(k,b_i,a_{-i}).
$$
For every $p\in\Delta(K\times\mathcal A_{-i})$, define
$$
L_i^+(p)
\coloneqq
\sum_{(k,B_{-i})\in K\times\mathcal A_{-i}}
p(k,B_{-i})
\max_{a_{-i}\in B_{-i}}d_i(k,a_{-i})
$$
and
$$
L_i^-(p)
\coloneqq
\sum_{(k,B_{-i})\in K\times\mathcal A_{-i}}
p(k,B_{-i})
\min_{a_{-i}\in B_{-i}}d_i(k,a_{-i}).
$$
These functionals are linear in $p$, and $L_i^-(p)\leq L_i^+(p)$.

For every $q\in\mathcal Q_i(p)$,
$$
\mathbb E_q
\left[
u_i(\cdot,a_i,\cdot)-u_i(\cdot,b_i,\cdot)
\right]
=
\sum_{k,B_{-i},a_{-i}}
q(k,B_{-i},a_{-i})d_i(k,a_{-i}).
$$
Compatibility restricts the conditional distribution of $a_{-i}$ at each $(k,B_{-i})$ only to have support in $B_{-i}$. Hence the maximum of this expression over $\mathcal Q_i(p)$ is $L_i^+(p)$, while its minimum is $L_i^-(p)$. Since player $i$ has only two actions,
$$
a_i\in\BR_i(p)
\quad\Longleftrightarrow\quad
L_i^+(p)\geq0,
$$
and
$$
b_i\in\BR_i(p)
\quad\Longleftrightarrow\quad
L_i^-(p)\leq0.
$$
It follows that
$$
\BR_i^{-1}(\{a_i\})
=
\{p:L_i^-(p)>0\},
$$
$$
\BR_i^{-1}(\{b_i\})
=
\{p:L_i^+(p)<0\},
$$
and
$$
\BR_i^{-1}(\{a_i,b_i\})
=
\{p:L_i^-(p)\leq0,\ L_i^+(p)\geq0\}.
$$
Each best-response region is therefore a mixed polyhedron. Hence the payoff structure is polyhedral, and thus $\BR$-convex.
\end{proof}

\subsection{Interval Payoff Structures}\label{sec:interval}

Suppose that every action set $A_i$ is totally ordered. For every
$B_i\in\mathcal A_i$, write $B_i$ as an \emph{interval} if
$$
B_i
=
\{a_i\in A_i:\underline a_i\leq a_i\leq\overline a_i\}
$$
for some $\underline a_i,\overline a_i\in A_i$.

For every player $i$ and every $p_i\in\Delta(K\times\mathcal A_{-i})$, define the maximal and minimal compatible joint distributions by
$$
\overline q_i(p_i)(k,B_{-i},a_{-i})
=
p_i(k,B_{-i})\mathbf 1_{\{a_{-i}=\max B_{-i}\}},
\qquad
\underline q_i(p_i)(k,B_{-i},a_{-i})
=
p_i(k,B_{-i})\mathbf 1_{\{a_{-i}=\min B_{-i}\}},
$$
where maxima and minima are taken coordinatewise.

\begin{definition}[Interval best-response sets]
The payoff structure has \emph{interval best-response sets} if, for every
player $i$ and every $p_i\in\Delta(K\times\mathcal A_{-i})$:
\begin{enumerate}[label=(\roman*)]
\item $\BR_i(p_i)$ is an interval;
\item
$$
\max\BR_i(p_i)
\in
\arg\max_{a_i\in A_i}
\mathbb E_{\overline q_i(p_i)}
u_i(\cdot,a_i,\cdot),
$$
and
$$
\min\BR_i(p_i)
\in
\arg\max_{a_i\in A_i}
\mathbb E_{\underline q_i(p_i)}
u_i(\cdot,a_i,\cdot).
$$
\end{enumerate}
\end{definition}

\begin{proposition}[Interval best-response sets]
If the payoff structure has interval best-response sets, then it is polyhedral, and thus $\BR$-convex.
\end{proposition}

\begin{proof}
Fix a player $i$ and an interval
$$
B_i
=
\{a_i\in A_i:\underline a_i\leq a_i\leq\overline a_i\}.
$$
We claim that $\BR_i(p_i)=B_i$ if and only if
$$
p_i\in\mathcal I_i(B_i),
$$
$$
\mathbb E_{\overline q_i(p_i)}
\left[
u_i(\cdot,\overline a_i,\cdot)
-
u_i(\cdot,a_i,\cdot)
\right]
>0
\quad\text{for every }a_i>\overline a_i,
$$
and
$$
\mathbb E_{\underline q_i(p_i)}
\left[
u_i(\cdot,\underline a_i,\cdot)
-
u_i(\cdot,a_i,\cdot)
\right]
>0
\quad\text{for every }a_i<\underline a_i.
$$

Suppose first that $\BR_i(p_i)=B_i$. By the extremal-distribution property, $\overline a_i$ is optimal under $\overline q_i(p_i)$. Every $a_i>\overline a_i$ must then yield a strictly lower payoff, since otherwise it would also belong to $\BR_i(p_i)$. The argument for actions below $\underline a_i$ is analogous.

Conversely, suppose that the three displayed conditions hold. Since $\BR_i(p_i)$ is an interval containing $B_i$, it can differ from $B_i$ only by containing an action above $\overline a_i$ or below $\underline a_i$. In the first case, $\max\BR_i(p_i)$ is optimal under $\overline q_i(p_i)$, contradicting the strict inequalities above $\overline a_i$. The second case is ruled out analogously by $\underline q_i(p_i)$. Hence $\BR_i(p_i)=B_i$.

By Lemma~\ref{Lem_BRpolytope}, the best-response inclusion region $\mathcal I_i(B_i)$ is a polyhedron. Each remaining condition gives a strict linear inequality in $p_i$. Thus $\BR_i^{-1}(B_i)$ is a mixed polyhedron. If $B_i$ is not an interval, then $\BR_i^{-1}(B_i)=\varnothing$. Hence the payoff structure is
polyhedral.
\end{proof}

We next give primitive conditions guaranteeing interval best-response sets.

Write
$$
A_i=\{a_i^1<\cdots<a_i^{R_i}\}.
$$

\begin{definition}[Own-action concavity]
The payoff function $u_i$ is \emph{concave in own action} if, for every
$(k,a_{-i})$ and every $r=1,\ldots,R_i-2$,
$$
u_i(k,a_i^{r+1},a_{-i})
-
u_i(k,a_i^r,a_{-i})
\geq
u_i(k,a_i^{r+2},a_{-i})
-
u_i(k,a_i^{r+1},a_{-i}).
$$
\end{definition}

\begin{definition}[Increasing differences]
The payoff function $u_i$ has \emph{increasing differences} if, for
every $a_i'\geq a_i$ and $a_{-i}'\geq a_{-i}$,
$$
u_i(k,a_i',a_{-i}')
-
u_i(k,a_i,a_{-i}')
\geq
u_i(k,a_i',a_{-i})
-
u_i(k,a_i,a_{-i})
$$
for every $k\in K$.
\end{definition}

\begin{proposition}[Concavity and increasing differences]
\label{prop:concavity-increasing-differences}
If every payoff function $u_i$ is concave in own action and has increasing differences, then the payoff structure has interval best-response sets.
\end{proposition}
Together with the preceding proposition, this result gives a primitive sufficient condition for polyhedrality and $\BR$-convexity. The proof is in Appendix~\ref{app:concavity-increasing-differences}.

These classes therefore admit direct representation by \ICRt\ hierarchies without augmentation. We now turn to arbitrary finite payoff structures.

\section{Direct Representation for General Payoff Structures}\label{sec:generalRP}

When best-response regions are nonconvex, coarsening types by their ordinary \ICRt\ hierarchies may average beliefs outside the announced regions. We restore convexity by refining each best-response region with finitely many cell tags. We first use the resulting finite \BR-extension to construct full augmented hierarchies. We then retain only the tags at transitions of the ordinary hierarchy to obtain a countable direct representation with finitely many obedience conditions per type.

\subsection{Full Augmented Hierarchies}
\label{sec:full-augmented}

\begin{definition}[Finite $\BR$-extension]
A \emph{finite $\BR$-extension} consists, for every player $i$, of a finite set $\mathcal C_i$ with a distinguished element $c_i^*$, a map $\alpha_i:\mathcal C_i\to\mathcal A_i$ satisfying $\alpha_i(c_i^*)=A_i$, and a measurable tagging map
$$
\tau_i:\Delta(K\times\mathcal A_{-i})\to\mathcal C_i
$$
such that $\alpha_i\circ\tau_i=\BR_i$. 
\end{definition}

The set $\mathcal C_i$ contains tags refining the action sets in the range of $\BR_i$, together with the distinguished tag $c_i^*$. The map $\tau_i$ assigns a tag to each belief, while $\alpha_i$ returns the associated best-response set. The condition $\alpha_i\circ\tau_i=\BR_i$ implies that $\BR_i$ is constant on every inverse image $\tau_i^{-1}(c_i)$, so these inverse images refine the best-response regions. The tag $c_i^*$ represents the level-zero action set $A_i$ and need not belong to the range of $\tau_i$.

\begin{definition}[$\BR$-convex and polyhedral extensions]
A finite $\BR$-extension is \emph{$\BR$-convex} if $\tau_i^{-1}(c_i)$ is convex for every player $i$ and every $c_i\in\mathcal C_i$. It is \emph{polyhedral} if $\tau_i^{-1}(c_i)$ is a mixed polyhedron for every player $i$ and every $c_i\in\mathcal C_i$.
\end{definition}

Every polyhedral $\BR$-extension is $\BR$-convex. Convex extensions make augmented hierarchies direct through the same averaging argument as under $\BR$-convex payoff structures, while polyhedral extensions additionally yield finite systems of weak and strict linear inequalities characterizing implementable distributions.

\begin{lemma}\label{lem:exist-extension}
Every finite payoff structure  admits a  polyhedral $\BR$-extension.
\end{lemma}

\begin{proof}
By the identity defining exact best-response regions and Lemma~\ref{Lem_BRpolytope}, each region $\BR_i^{-1}(B_i)$ is the difference between the polyhedron $\mathcal I_i(B_i)$ and a finite union of polyhedra $\mathcal I_i(B_i\cup\{a_i\})$. Such a difference admits a finite disjoint decomposition into mixed polyhedra: partition the complement of each subtracted polyhedron according to the first defining inequality it violates, and intersect the resulting finite partitions. After discarding empty sets, choose mixed polyhedra $(D_i^r(B_i))_{r=1}^{R_i(B_i)}$ such that
$$
\BR_i^{-1}(B_i)=\bigsqcup_{r=1}^{R_i(B_i)}D_i^r(B_i).
$$
For every such set $D_i^r(B_i)$, introduce a tag $c_i^r(B_i)$, and adjoin a distinguished tag $c_i^*$. Let
$$
\mathcal C_i
=
\{c_i^*\}
\cup
\{c_i^r(B_i):B_i\in\mathcal A_i,\ 1\leq r\leq R_i(B_i)\}.
$$
Define $\alpha_i(c_i^*)=A_i$ and $\alpha_i(c_i^r(B_i))=B_i$. Since the best-response regions partition $\Delta(K\times\mathcal A_{-i})$, the sets $D_i^r(B_i)$ form a finite partition of this simplex. Hence the map $\tau_i:\Delta(K\times\mathcal A_{-i})\to\mathcal C_i$ defined by
$$
\tau_i(p)=c_i^r(B_i)
\quad\Longleftrightarrow\quad
p\in D_i^r(B_i)
$$
is well defined. Moreover, if $p\in D_i^r(B_i)$, then $\BR_i(p)=B_i=\alpha_i(\tau_i(p))$. Thus $\alpha_i\circ\tau_i=\BR_i$, so $(\mathcal C_i,\alpha_i,\tau_i)_i$ is a finite $\BR$-extension.

Finally, $\tau_i^{-1}(c_i^r(B_i))=D_i^r(B_i)$ is a mixed polyhedron, while $\tau_i^{-1}(c_i^*)=\varnothing$, which is also a mixed polyhedron. Hence the extension is polyhedral.
\end{proof}

Figure~\ref{fig:tagged-decomposition} applies this construction to the nonconvex best-response region in Figure~\ref{fig:nonconvex-br}. Its two components receive distinct tags with the same best-response-set label.

\begin{figure}[htbp]
\centering
\begin{tikzpicture}[x=.92cm,y=.82cm]
  \coordinate (P) at (0,0);
  \coordinate (Pp) at (7,0);
  \coordinate (R) at (3.5,4.2);

  \fill[blue!14] (0,0)--(2.33,0)--(1.17,1.4)--cycle;
  \fill[teal!16] (7,0)--(5.83,1.4)--(4.67,0)--cycle;
  \fill[red!16]
    (3.5,4.2)--(5.83,1.4)--(4.67,0)--(2.33,0)--(1.17,1.4)--cycle;

  \draw[brown!75!black,very thick] (P)--(Pp)--(R)--cycle;
  \draw[blue!70!black,thick] (0,0)--(2.33,0)--(1.17,1.4)--cycle;
  \draw[teal!70!black,thick] (7,0)--(5.83,1.4)--(4.67,0)--cycle;
  \draw[red!65!black,very thick]
    (3.5,4.2)--(5.83,1.4)--(4.67,0)--(2.33,0)--(1.17,1.4)--cycle;

  \fill[blue!70!black] (P) circle (2.7pt)
    node[below left,font=\scriptsize] {$p$};
  \fill[teal!70!black] (Pp) circle (2.7pt)
    node[below right,font=\scriptsize] {$p'$};
  \fill[blue!70!black] (R) circle (2.7pt)
    node[above,font=\scriptsize] {$r$};
\end{tikzpicture}
\caption{A polyhedral $\BR$-extension of Example 2. The triangle in Figure~\ref{fig:nonconvex-br} is partitioned into three convex cells. The blue and teal zones receive distinct tags, both of which project to the best-response set $\{a,b\}$. The red zone receives a tag that projects to $\{a,b,c\}$. Thus the tags distinguish the two components of the nonconvex best-response region.}
\label{fig:tagged-decomposition}
\end{figure}
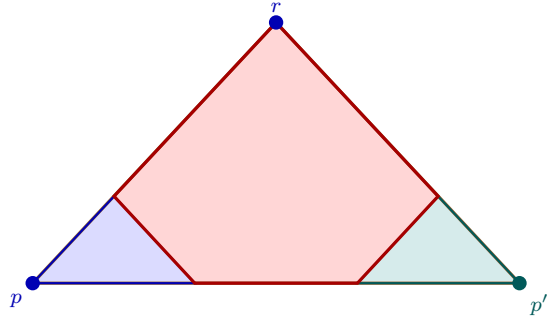

\paragraph{Augmented hierarchies.}
A finite $\BR$-extension $(\mathcal C_i,c_i^*,\alpha_i,\tau_i)_i$  induces an operator $\beta_i:\Delta(K\times\mathcal C_{-i})\to\mathcal C_i$ on tagged beliefs, defined by
$$
\beta_i(q)
\coloneqq
\tau_i\bigl((\id\times\alpha_{-i})_\#q\bigr).
$$
Since $\alpha_i\circ\tau_i=\BR_i$, this operator satisfies
$$
\alpha_i\circ\beta_i
=
\BR_i\circ(\id\times\alpha_{-i})_\#.
$$
For every $c_i\in\mathcal C_i$,
$$
\beta_i^{-1}(c_i)
=
\left((\id\times\alpha_{-i})_\#\right)^{-1}
\bigl(\tau_i^{-1}(c_i)\bigr).
$$
Hence $\beta_i^{-1}(c_i)$ is convex when the extension is $\BR$-convex and is a mixed polyhedron when the extension is polyhedral.

Consider a common prior $P\in\Delta(K\times T)$, with conditional probabilities $P_i:T_i\to\Delta(K\times T_{-i})$. The \emph{augmented hierarchy} $(\AICR_i^m(t_i))_{m\geq0}$ of a type $t_i\in T_i$ is defined recursively by
\begin{enumerate}
\item[(i)] $\AICR_i^0(t_i)\coloneqq c_i^*$;
\item[(ii)] for every $m\in\naturals$,
\begin{equation}\label{ICR2}
\AICR_i^{m+1}(t_i)
\coloneqq
\beta_i\left(
(\id\times\AICR_{-i}^m)_\#P_i(t_i)
\right).
\end{equation}
\end{enumerate}

The map $\beta_i$ is measurable because $\tau_i$ is measurable and the pushforward map induced by $\id\times\alpha_{-i}$ is continuous between finite-dimensional probability simplices. It follows by induction from \eqref{ICR2} that every map $\AICR_i^m\colon T_i\to\mathcal C_i$ is measurable. Hence the full hierarchy map $\AICR_i\colon T_i\to\mathcal C_i^{\naturals}$ is measurable with respect to the product Borel structure.

Every common prior $P$ therefore induces a well-defined distribution
$$
P_{\AICR}
\coloneqq
(\id\times\AICR)_\#P
$$
on $K\times\mathcal C^{\naturals}$. Let $\mathcal P'$ denote the set of distributions obtained in this way as the common prior and the type spaces vary.

Unlike ordinary \ICRt\,hierarchies, augmented hierarchies need not form a countable space. Distinct tags may project onto the same action set, and the tag may change infinitely often even though the projected hierarchy has only finitely many strict decreases.

The map $\alpha$ recovers the ordinary \ICRt\,hierarchy from the augmented hierarchy.

\begin{proposition}\label{prop:AICR-projects}
For every common prior $P$, the augmented hierarchy projects level by level to the ordinary hierarchy: for every player $i$, every type $t_i$, and every $m\in\naturals$,
$$
\alpha_i\bigl(\AICR_i^m(t_i)\bigr)
=
\ICR_i^m(t_i).
$$
\end{proposition}

\begin{proof}
The equality holds at level zero because $\alpha_i(c_i^*)=A_i=\ICR_i^0(t_i)$. Suppose it holds at level $m$ for every player. By \eqref{ICR2} and the definition of $\beta_i$,
\begin{align*}
\alpha_i\bigl(\AICR_i^{m+1}(t_i)\bigr)
&=
\BR_i\left(
\bigl(\id\times(\alpha_{-i}\circ\AICR_{-i}^m)\bigr)_\#P_i(t_i)
\right)\\
&=
\BR_i\left(
(\id\times\ICR_{-i}^m)_\#P_i(t_i)
\right)
=
\ICR_i^{m+1}(t_i).
\end{align*}
The result follows by induction.
\end{proof}

The induced distribution $P_{\AICR}$ may itself be viewed as a common prior whose type space for player $i$ is contained in $\mathcal C_i^{\naturals}$. Directness requires the generated augmented hierarchy to equal the announced one. We now characterize the distributions satisfying this identity condition.

\begin{theorem}[Direct Representation by Augmented Hierarchies]\label{RPGeneral}
Fix a finite $\BR$-convex extension. For every $P\in\Delta(K\times\mathcal C^{\naturals})$, the following conditions are equivalent:
\begin{enumerate}
\item $P\in\mathcal P'$;
\item When viewed as a common prior on augmented hierarchy types, $P$ satisfies $\AICR(c)=c$ almost surely;
\item $P(c_i^0=c_i^*)=1$ and, for every player $i$ and every $m\geq1$,
\begin{equation}\label{OC2}
c_i^m
=
\beta_i\left(
\marg_{k,c_{-i}^{m-1}}P(\cdot,\cdot\mid c_i)
\right)
\quad
\text{for almost every $c_i$.}
\end{equation}
\end{enumerate}
\end{theorem}

\begin{proof}
The proof follows that of Theorem~\ref{RP1}, with $\beta_i$ and augmented hierarchies replacing $\BR_i$ and ordinary hierarchies. For $(1)\Rightarrow(3)$, conditional on $c_i$, every original level-$m$ belief belongs to the convex cell $\beta_i^{-1}(c_i^m)$. Their conditional average, which is the direct belief in \eqref{OC2}, remains in this cell. The level-zero condition holds by construction.

For $(3)\Rightarrow(2)$, regard $P$ as a common prior on the announced augmented hierarchies. The level-zero condition and \eqref{OC2} imply by induction that the generated hierarchy equals the announced one. For $(2)\Rightarrow(1)$, use $P$ itself as the common prior. Its induced distribution over states and augmented hierarchies is $P$.
\end{proof}

The theorem makes augmented hierarchies direct types for every finite payoff structure, while condition \emph{3} gives the obedience characterization.

For a polyhedral extension, each cell $\beta_i^{-1}(c_i)$ is defined by one finite system of weak and strict linear inequalities. Hence each augmented obedience constraint is one such system, rather than a disjunction over cells. These inequalities apply to regular conditional beliefs and remain meaningful when the marginal distribution over augmented hierarchies is nonatomic. Since every finite payoff structure admits a polyhedral $\BR$-extension, every finite payoff structure admits this tagged conditional-belief characterization.

\paragraph{The size of full augmentation.}
The full augmented hierarchy resolves the geometric obstruction to directness, but it records a cell tag at every round. Its type space may therefore be uncountable, and its obedience characterization may require countably many conditions per type.

As observed in Section~\ref{sec:brconvex}, the set $S_i$ of ordinary \ICRt\,hierarchies is countable. Full augmented hierarchies need not form a countable set. Even after the ordinary action set becomes constant, its tag may move between distinct cells carrying the same action-set label. Arbitrary infinite tag sequences may then be realized without changing the ordinary hierarchy.

The following continuation of Example~2 illustrates both the cell decomposition used by full augmentation and the source of its uncountable type space.

\paragraph{Example 2 continued.}
\label{sec:uncountable-augmented}
Consider the payoff structure of Example 2. Let $p^0$ assign probability one to $\{x\}$, let $p^y$ assign probability one to $\{x,y\}$, and let $p^z$ assign probability one to $\{x,z\}$. Consider beliefs
$$
q=\alpha p^0+\beta p^y+\gamma p^z,
\qquad
\alpha+\beta+\gamma=1,
$$
with $\alpha,\beta,\gamma\geq0$. The compatible distribution that always assigns action $x$ makes both $a$ and $b$ optimal. Hence $a,b\in\BR_1(q)$ throughout this triangle.

Suppose that a compatible joint distribution induces probabilities $(X,Y,Z)$ over $(x,y,z)$. Since
$$
U_1(a)=1-Z,
\qquad
U_1(b)=1-Y,
\qquad
U_1(c)=\frac23,
$$
action $c$ is optimal if and only if $Y,Z\geq1/3$. Under $q$, the maximal probabilities that can be assigned to $y$ and $z$ are $\beta$ and $\gamma$, respectively, and these bounds can be attained simultaneously. Therefore
$$
c\in\BR_1(q)
\quad\Longleftrightarrow\quad
\beta\geq\frac13
\ \text{ and }\
\gamma\geq\frac13.
$$
Consequently,
$$
\BR_1(q)=\{a,b\}
\quad\Longleftrightarrow\quad
\beta<\frac13
\ \text{ or }\
\gamma<\frac13.
$$
This region is connected because its two parts are convex and intersect at $p^0$, but it is not convex because it contains $p^y$ and $p^z$ but not their midpoint.

One possible disjoint convex-cell decomposition of the triangle is
$$
D^L
=
\left\{q:\gamma<\frac13,\ \beta\geq\gamma\right\},
\qquad
D^R
=
\left\{q:\beta<\frac13,\ \gamma>\beta\right\},
$$
and
$$
D^c
=
\left\{q:\beta\geq\frac13,\ \gamma\geq\frac13\right\}.
$$
The cells $D^L$ and $D^R$ both have best-response-set label $\{a,b\}$, while $D^c$ has label $\{a,b,c\}$. Let
$$
v=\frac13p^0+\frac13p^y+\frac13p^z
$$
be the interior vertex of $D^c$. The segment $[p^0,v)$ is assigned to $D^L$; assigning it to $D^R$ instead would give another valid decomposition.

\begin{figure}[htbp]
\centering
\begin{minipage}[c]{0.48\textwidth}
\centering
\begin{tikzpicture}[x=.86cm,y=.77cm]
  \coordinate (Py) at (0,0);
  \coordinate (Pz) at (7,0);
  \coordinate (P0) at (3.5,4.2);
  \coordinate (L) at (2.33,0);
  \coordinate (R) at (4.67,0);
  \coordinate (V) at (3.5,1.4);

  \fill[brown!16] (Py)--(Pz)--(P0)--cycle;
  \draw[brown!75!black,very thick] (Py)--(Pz)--(P0)--cycle;

  \fill[red!16] (L)--(R)--(V)--cycle;
  \draw[red!65!black,very thick] (L)--(R)--(V)--cycle;

  \fill[blue!70!black] (Py) circle (2.7pt)
    node[below left,font=\scriptsize] {$p^y$};
  \fill[blue!70!black] (Pz) circle (2.7pt)
    node[below right,font=\scriptsize] {$p^z$};
  \fill[blue!70!black] (P0) circle (2.7pt)
    node[above,font=\scriptsize] {$p^0$};
\end{tikzpicture}

\smallskip
{\scriptsize Best-response regions}
\end{minipage}
\hfill
\begin{minipage}[c]{0.48\textwidth}
\centering
\begin{tikzpicture}[x=.86cm,y=.77cm]
  \coordinate (Py) at (0,0);
  \coordinate (Pz) at (7,0);
  \coordinate (P0) at (3.5,4.2);
  \coordinate (L) at (2.33,0);
  \coordinate (R) at (4.67,0);
  \coordinate (V) at (3.5,1.4);

  \fill[blue!14] (Py)--(L)--(V)--(P0)--cycle;
  \fill[green!16] (Pz)--(R)--(V)--(P0)--cycle;
  \fill[red!16] (L)--(R)--(V)--cycle;

  \draw[blue!70!black,thick] (Py)--(L)--(V)--(P0)--cycle;
  \draw[green!50!black,thick] (Pz)--(R)--(V)--(P0)--cycle;
  \draw[red!65!black,very thick] (L)--(R)--(V)--cycle;
  \draw[blue!70!black,very thick] (P0)--(V);

  \fill[blue!70!black] (Py) circle (2.7pt)
    node[below left,font=\scriptsize] {$p^y$};
  \fill[green!50!black] (Pz) circle (2.7pt)
    node[below right,font=\scriptsize] {$p^z$};
  \fill[blue!70!black] (P0) circle (2.7pt)
    node[above,font=\scriptsize] {$p^0$};
  \fill[red!70!black] (V) circle (2.7pt)
    node[right,font=\scriptsize] {$v$};
\end{tikzpicture}

\smallskip
{\scriptsize A disjoint convex-cell decomposition}
\end{minipage}
\caption{The triangle and a disjoint convex-cell decomposition. Left: the brown region has best-response set $\{a,b\}$ and the red region has best-response set $\{a,b,c\}$. Right: the blue cell $D^L$ and the green cell $D^R$ both project to $\{a,b\}$, while the red cell $D^c$ projects to $\{a,b,c\}$. The segment $[p^0,v)$ is assigned to $D^L$, while $v$ belongs to $D^c$.}
\label{fig:connected-nonconvex-br}
\end{figure}
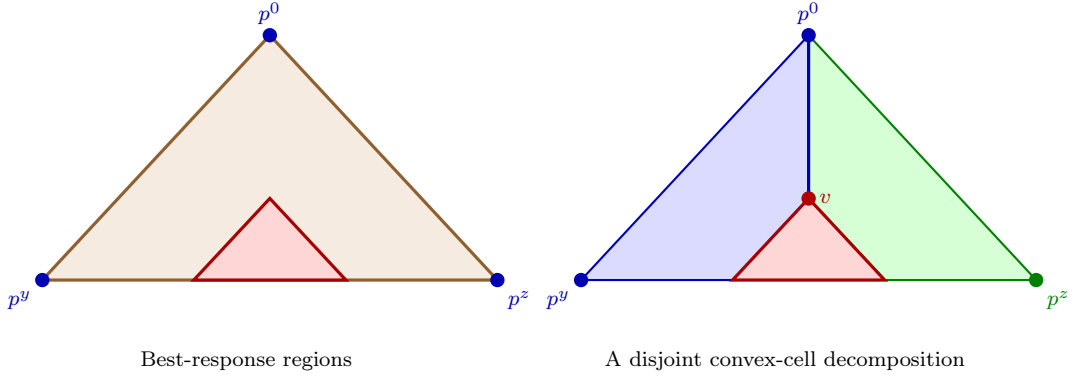

One can construct a state space in which, while the ordinary best-response set remains $\{a,b\}$, every sequence of visits to the blue and green cells is realized by some type. The set of such sequences has the cardinality of the continuum. Thus full augmented hierarchies can genuinely be uncountable.

\subsection{Transition-Tagged Hierarchies}
\label{sec:transition-tagged}
\label{sec:exact-countability}

Full augmentation records every cell visited by the belief sequence, including cell changes that occur while the ordinary best-response set remains constant. The ordinary hierarchy already records the successive best-response sets and the rounds at which they change. Moreover, Lemma~\ref{lem:transition-obedience-ordinary} shows that obedience throughout a constant stretch follows from obedience at its entry and exit. We therefore retain only the cell tags at those rounds.

Fix a finite $\BR$-convex extension $(\mathcal C_i,c_i^*,\alpha_i,\tau_i)_i$. A \emph{transition-tagged hierarchy} of player $i$ consists of an ordinary hierarchy $s_i\in S_i$, an incoming tag
$$
c_i^{\mathrm{in}}(B_i)\in\mathcal C_i
\quad\text{for every }B_i\in\mathcal R_i(s_i),
$$
and an outgoing tag
$$
c_i^{\mathrm{out}}(B_i)\in\mathcal C_i
\quad\text{whenever }\overline m_i(B_i;s_i)<\infty,
$$
such that
$$
\alpha_i(c_i^{\mathrm{in}}(B_i))=B_i,
\qquad
\alpha_i(c_i^{\mathrm{out}}(B_i))=B_i
$$
whenever the corresponding tag is defined. Write
$$
z_i=(s_i,c_i^{\mathrm{in}},c_i^{\mathrm{out}})
$$
and denote the set of transition-tagged hierarchies by $Z_i$.

The space $Z_i$ is countable. The ordinary hierarchy space $S_i$ is countable, every ordinary hierarchy contains at most $|A_i|$ plateaus, and each plateau carries at most two tags drawn from the finite set $\mathcal C_i$.

Consider $P\in\Delta(K\times Z)$ as a common prior under which player $i$ observes $z_i$. Define
$$
p_i^m(z_i)
\coloneqq
\marg_{k,s_{-i}^{m-1}}
P(\cdot,\cdot\mid z_i),
\qquad
p_i^\infty(z_i)
\coloneqq
\marg_{k,s_{-i}^\infty}
P(\cdot,\cdot\mid z_i).
$$

Every common prior $Q\in\Delta(K\times T)$ induces a transition-tagged hierarchy. Write $s_i(t_i)=\ICR_i(t_i)$ and
$$
q_i^m(t_i)
=
(\id\times\ICR_{-i}^{m-1})_\#Q_i(t_i).
$$
Applying $\tau_i$ to the belief at the entry and every finite exit of each plateau defines a map $\zeta_i:T_i\to Z_i$. Let $\zeta=(\zeta_i)_i$ and denote the induced distribution by
$$
Q_Z
=
(\id\times\zeta)_\#Q.
$$

\begin{theorem}[Direct Representation by Transition-Tagged Hierarchies]
\label{thm:plateau-RP}
Fix a finite $\BR$-convex extension. For every $P\in\Delta(K\times Z)$, the following conditions are equivalent:
\begin{enumerate}
\item There exists a common prior $Q$ such that $P=Q_Z$.
\item When viewed as a common prior on $Z$, $P$ generates the announced ordinary hierarchies and reproduces the announced transition tags.
\item For every player $i$ and almost every $z_i=(s_i,c_i^{\mathrm{in}},c_i^{\mathrm{out}})$:
\begin{enumerate}
\item for every $B_i\in\mathcal R_i(s_i)$,
$$
\tau_i\left(
p_i^{\underline m_i(B_i;s_i)}(z_i)
\right)
=
c_i^{\mathrm{in}}(B_i);
$$
\item for every $B_i\in\mathcal R_i(s_i)$ such that $\overline m_i(B_i;s_i)<\infty$,
$$
\tau_i\left(
p_i^{\overline m_i(B_i;s_i)}(z_i)
\right)
=
c_i^{\mathrm{out}}(B_i);
$$
\item for the terminal plateau,
$$
s_i^\infty
\subseteq
\BR_i\bigl(p_i^\infty(z_i)\bigr).
$$
\end{enumerate}
\end{enumerate}
\end{theorem}

\begin{proof}
$(1)\Rightarrow(3)$. Conditional on $z_i$, all original beliefs at a recorded transition belong to the cell identified by the corresponding tag. The cell is convex, so their conditional average belongs to the same cell and reproduces the tag. For the terminal condition, the original beliefs $q_i^m(t_i)$ converge in total variation to the belief induced by the terminal opponent sets. Every action in $s_i^\infty$ is a best response at every finite round and remains so at the limit by Lemma~\ref{Lem_BRpolytope}(i). The best-response inclusion region $\mathcal I_i(s_i^\infty)$ is convex by Lemma~\ref{Lem_BRpolytope}(ii), so the terminal inclusion is preserved under conditional averaging.

$(3)\Rightarrow(2)$. Projecting the incoming and outgoing tag conditions through $\alpha_i$ gives the entry and finite-exit conditions of Lemma~\ref{lem:transition-obedience-ordinary}, conditional on $z_i$. The same monotonicity argument therefore gives ordinary obedience at every round, and the endpoint conditions reproduce the announced transition tags.

$(2)\Rightarrow(1)$. Use $P$ itself as the common prior on $Z$. Its induced distribution of transition-tagged hierarchies is $P$.
\end{proof}

The theorem shows that transition tagging preserves directness while discarding tag changes within plateaus.

As in the ordinary case, a transition-tagged hierarchy with $q_i$ plateaus requires $q_i$ incoming-tag conditions, $q_i-1$ outgoing-tag conditions, and one terminal condition. Thus the general direct representation also admits a description with at most $2|A_i|$ nontrivial obedience conditions per type.

The preceding theorem combines the three properties sought in this section: transition-tagged hierarchies are direct, their type spaces are countable, and their induced distributions are characterized by finitely many obedience conditions per type. Projecting away the tags gives the following exact representation of ordinary \ICRt\ hierarchies.

\begin{corollary}[Exact Countable Direct Representation]
\label{cor:countable-RP}
For every finite payoff structure and every common prior, there exists a common prior with countable type spaces that induces the same joint distribution over the state and the entire ordinary \ICRt\,hierarchy.
\end{corollary}

\begin{proof}
By Lemma~\ref{lem:exist-extension}, the payoff structure admits a finite polyhedral, hence $\BR$-convex, extension. Apply Theorem~\ref{thm:plateau-RP} to the induced distribution of transition-tagged hierarchies. This distribution is itself a common prior on the countable type spaces $Z_i$, and its marginal distribution over the state and ordinary hierarchies coincides with that induced by the original common prior.
\end{proof}

The exact direct representation records the full timing of the ordinary hierarchy but only finitely many tags for each type. It retains the unbounded transition dates required by electronic-mail examples while eliminating the infinite cell crossings that make the full augmented hierarchy space uncountable.

\section{Finite Approximation}
\label{sec:finite-approximation}

Countability is sharp for exact representation. In Example~1, the common prior assigns positive probability to infinitely many distinct \ICRt\ hierarchies with unbounded transition dates. No finite type space can reproduce that state--hierarchy distribution exactly. We therefore turn from exact representation of entire hierarchies to approximation of terminal \ICRt\ outcome distributions. These distributions belong to the finite-dimensional simplex $\Delta(K\times\mathcal A)$, on which total variation induces the usual Euclidean topology.

For a common prior $P$, let
$$
\nu_P
\coloneqq
(\id\times\ICR^\infty)_\#P
\in\Delta(K\times\mathcal A)
$$
denote its terminal \ICRt\,outcome distribution.

The construction is related to \citet{lipman2003finite}, who shows that, in finite partitions models, any fixed finite order of beliefs and knowledge consistent with common support can be matched by a finite common-prior model. Our truncation retains a payoff-dependent record instead, including the terminal \ICRt\ set, because agreement at finitely many belief orders does not control terminal \ICRt\ behavior.

For approximation, it is useful to return to full augmented hierarchies and retain only a finite prefix. Fix a common prior $P\in\Delta(K\times T)$ and a finite polyhedral $\BR$-extension. Its finite $M$-truncation, denoted $Q^M$, replaces each type $t_i$ with the record consisting of its terminal \ICRt\ set and its augmented tags through round $M$. The formal construction is given in Appendix~\ref{app:finite-approximation}.

\begin{theorem}[Finite-type approximation by truncations]
\label{thm:finite-type-approximation}
The terminal \ICRt\,outcome distributions of $P$ and its $M$-truncation satisfy
$$
\|\nu_{Q^M}-\nu_P\|_{\mathrm{TV}}
\longrightarrow0.
$$
For every $M$, the augmented-hierarchy distribution induced by $Q^M$ satisfies the identity condition~(2) of Theorem~\ref{RPGeneral}, has finite support, and is concentrated on eventually constant augmented hierarchies.
\end{theorem}

The truncation records the terminal \ICRt\ set and the first $M$ augmented tags. Convexity of the tagged cells reproduces the recorded prefix, while monotonicity ensures that the declared terminal actions remain rationalizable. The approximation error is bounded by the probability that some original hierarchy has not reached its terminal \ICRt\ set by round $M$, which converges to zero. The proof is given in Appendix~\ref{app:finite-approximation}.

The approximating common priors can also be selected from a fixed countable collection.

\begin{theorem}[Countable outcome-dense family of types]
\label{thm:countable-dense-types}
For every finite payoff structure, there exist a countable collection $\mathbb Q$ of finite common priors and a countable family $\mathbb T_i$ of Harsanyi types for each player $i$ such that every common prior in $\mathbb Q$ uses only types in $(\mathbb T_i)_i$ and, for every common prior $P$, there exists a sequence $(Q^n)_n$ in $\mathbb Q$ satisfying
$$
\|\nu_{Q^n}-\nu_P\|_{\mathrm{TV}}
\longrightarrow0.
$$
\end{theorem}

At every truncation level, the graph associating a finite common prior with its terminal \ICRt\,distribution is a subspace of a finite-dimensional Euclidean space and is therefore separable. Taking a countable dense subset at each level and then the union over levels yields the theorem. The detailed argument is given in Appendix~\ref{app:countable-dense}.

\section{Conclusion}\label{sec:conclusion}

For a fixed payoff structure, an information structure projects each type into the output language of \ICRt. Terminal \ICRt\ sets are generally too coarse because they omit the lower-order eliminations that generate them. Under \BR-convexity, full \ICRt\ hierarchies are direct types: recomputing \ICRt\ on the projected model yields the identity. The induced state--hierarchy distributions are characterized by this identity condition and level-by-level obedience. For polyhedral payoff structures, obedience admits explicit linear representations; binary-action structures and ordered structures with interval best-response sets provide natural classes.

For arbitrary finite payoff structures, finite polyhedral \BR-extensions refine best-response regions with finitely many tags. The resulting augmented hierarchies are direct under the augmented best-response operator, while their projections recover the ordinary \ICRt\ hierarchy. Thus every finite payoff structure admits an exact payoff-dependent direct representation.

The size of this representation can be reduced in two directions. Obedience need only be checked at the entries and exits of constant-action-set plateaus, so every ordinary or transition-tagged type requires at most $2|A_i|$ nontrivial conditions. Full augmented hierarchies may form a continuum because their tags can alternate indefinitely while the ordinary hierarchy remains constant. Transition tags discard these intermediate cell crossings and yield an exact countable representation. This cardinality is sharp: the electronic-mail example induces infinitely many ordinary hierarchies with distinct transition dates and therefore cannot be represented exactly on a finite type space. Finite type spaces are nevertheless dense in terminal \ICRt\ outcome distributions.

Table~\ref{tab:cardinality-comparison} summarizes the represented object, scope, cardinality, and number of obedience constraints (OCs) per type for each construction.

\begin{table}[htbp]
\centering
\tiny
\renewcommand{\arraystretch}{1.25}
\begin{tabular}{p{0.23\textwidth}p{0.25\textwidth}p{0.19\textwidth}p{0.14\textwidth}}
\hline
Representation
& Represented object and scope
& Type-space cardinality
& OCs per type\\
\hline
Ordinary hierarchy
& State--ordinary-hierarchy distributions; exact in $\BR$-convex games
& Countable
& Finite ($\leq 2|A_i|$)\\
Full augmented hierarchy
& State--augmented-hierarchy distributions; exact in general games and projecting onto ordinary hierarchies
& Possibly continuum
& Countable\\
Transition-tagged hierarchy
& State--transition-tagged distributions; exact in general games and projecting onto ordinary hierarchies
& Countable
& Finite ($\leq 2|A_i|$)\\
Finite-type approximation
& State--terminal-\ICRt-set distributions; approximate in general games
& Finite
& Finite\\
\hline
\end{tabular}
\caption{Summary of direct-representation and approximation results for finite games.}
\label{tab:cardinality-comparison}
\end{table}

For binary-action supermodular games, \citet{morris2020implementation} go further by characterizing the closure of state--action distributions induced in the smallest equilibrium and using this characterization for information design. Our results provide an exact direct representation of state--\ICRt\ hierarchy distributions for arbitrary finite games, but do not characterize their projection onto rationalizable state--action distributions. We leave that further step for future work.

\newpage
\bibliographystyle{ecta}

\bibliography{DRep}

\appendix

\section{Proof of the Best-Response Geometry Lemma}
\label{app:br-geometry}

\begin{proof}[Proof of Lemma~\ref{Lem_BRpolytope}]
\noindent\textbf{(i).}
Fix $a_i\in A_i$. The condition $a_i\in\BR_i(p)$ is equivalent to the existence of nonnegative variables $q(k,B_{-i},a_{-i})$ such that
$$
q(k,B_{-i},a_{-i})=0 \quad\text{if }a_{-i}\notin B_{-i},
$$
$$
\sum_{a_{-i}\in B_{-i}}q(k,B_{-i},a_{-i})=p(k,B_{-i})
\quad\text{for every }(k,B_{-i}),
$$
and, for every $a_i'\in A_i$,
$$
\sum_{k,B_{-i},a_{-i}}
\bigl(u_i(k,a_i,a_{-i})-u_i(k,a_i',a_{-i})\bigr)
q(k,B_{-i},a_{-i})\geq 0.
$$
These conditions form a finite system of weak affine inequalities and equalities in $(p,q)$. By Fourier--Motzkin elimination \citep[Section~12.2, pp.~155--156]{Schrijver1986}, eliminating the variables $q$ yields an equivalent finite system of weak affine inequalities in $p$. Hence $\mathcal I_i(\{a_i\})$ is a polyhedron.

\noindent\textbf{(ii).}
For every $B_i\in\mathcal A_i$,
$$
\mathcal I_i(B_i)
=
\bigcap_{a_i\in B_i}\mathcal I_i(\{a_i\}),
$$
so (ii) follows from (i).

\noindent\textbf{(iii).}
Recall that
$$
\BR_i^{-1}(B_i)
=
\mathcal I_i(B_i)
\setminus
\bigcup_{a_i\in A_i\setminus B_i}
\mathcal I_i(B_i\cup\{a_i\}).
$$
By (ii), each best-response inclusion region in this expression is defined by finitely many weak affine inequalities. The complement of such a region is a finite union of sets defined by one strict affine inequality. Distributing finite intersections over finite unions expresses $\BR_i^{-1}(B_i)$ as a finite union of mixed polyhedra.
\end{proof}

\section{Proof of the Interval Best-Response Result}

\subsection{Concavity and Increasing Differences}
\label{app:concavity-increasing-differences}

\begin{proof}[Proof of Proposition~\ref{prop:concavity-increasing-differences}]
Fix a player $i$ and a belief $p_i\in\Delta(K\times\mathcal A_{-i})$. For every $q_i\in\mathcal Q_i(p_i)$, let
$$
V_{q_i}(a_i)
\coloneqq
\mathbb E_{q_i}
u_i(\cdot,a_i,\cdot)
$$
and, for $r=1,\ldots,R_i-1$, let
$$
D_{q_i}^r
\coloneqq
V_{q_i}(a_i^{r+1})-V_{q_i}(a_i^r).
$$
Own-action concavity implies
$$
D_{q_i}^1\geq\cdots\geq D_{q_i}^{R_i-1},
$$
so the set of maximizers of $V_{q_i}$ is an interval.

We next show that the union of these maximizer sets over compatible distributions is an interval. Suppose that $a_i^r$ is optimal under $q_i^0\in\mathcal Q_i(p_i)$ and $a_i^s$ is optimal under $q_i^1\in\mathcal Q_i(p_i)$, where $r<t<s$. Optimality and decreasing increments imply
$$
D_{q_i^0}^t\leq0
\qquad\text{and}\qquad
D_{q_i^1}^{t-1}\geq0.
$$
For $\lambda\in[0,1]$, let
$$
q_i^\lambda
=
(1-\lambda)q_i^0+\lambda q_i^1.
$$
Then $q_i^\lambda\in\mathcal Q_i(p_i)$, and the increments $D_{q_i^\lambda}^j$ are affine in $\lambda$.

If $D_{q_i^0}^{t-1}\geq0$, then
$$
D_{q_i^0}^{t-1}\geq0\geq D_{q_i^0}^t.
$$
Since the increments are decreasing, all increments before $t$ are weakly positive and all increments from $t$ onward are weakly negative. Hence $a_i^t$ is optimal under $q_i^0$. Otherwise, there exists $\lambda\in(0,1]$ such that $D_{q_i^\lambda}^{t-1}=0$. Since the increments are decreasing,
$$
D_{q_i^\lambda}^t
\leq
D_{q_i^\lambda}^{t-1}
=
0.
$$
Again, all increments before $t$ are weakly positive and all increments from $t$ onward are weakly negative. Hence $a_i^t$ is optimal under $q_i^\lambda$. Therefore $\BR_i(p_i)$ is an interval.

It remains to prove the extremal-distribution property. Fix any $q_i\in\mathcal Q_i(p_i)$. Under $q_i$, the action profile $a_{-i}$ is coordinatewise below $\max B_{-i}$ almost surely. Since $q_i$ and $\overline q_i(p_i)$ have the same marginal on $(k,B_{-i})$, increasing differences imply that, for every $a_i'\geq a_i$,
$$
\mathbb E_{\overline q_i(p_i)}
\left[
u_i(\cdot,a_i',\cdot)-u_i(\cdot,a_i,\cdot)
\right]
\geq
\mathbb E_{q_i}
\left[
u_i(\cdot,a_i',\cdot)-u_i(\cdot,a_i,\cdot)
\right].
$$
Let $\overline a_i$ be the greatest maximizer under $\overline q_i(p_i)$. If some $a_i>\overline a_i$ were optimal under $q_i$, then
$$
\mathbb E_{\overline q_i(p_i)}
\left[
u_i(\cdot,a_i,\cdot)-u_i(\cdot,\overline a_i,\cdot)
\right]
\geq
\mathbb E_{q_i}
\left[
u_i(\cdot,a_i,\cdot)-u_i(\cdot,\overline a_i,\cdot)
\right]
\geq0.
$$
Since $\overline a_i$ is optimal under $\overline q_i(p_i)$, $a_i$ must also be optimal under $\overline q_i(p_i)$, contradicting the definition of $\overline a_i$ as its greatest maximizer. Thus every action in $\BR_i(p_i)$ is weakly below $\overline a_i$. Since $\overline a_i\in\BR_i(p_i)$, it follows that
$$
\max\BR_i(p_i)=\overline a_i.
$$
The analogous argument using $\underline q_i(p_i)$ shows that $\min\BR_i(p_i)$ is a best response under the minimal compatible distribution. Hence the payoff structure has interval best-response sets.
\end{proof}

\section{Proofs for the Approximation Results}

\subsection{Finite-Type Approximation}
\label{app:finite-approximation}

We prove Theorem~\ref{thm:finite-type-approximation}. Fix a common prior $P\in\Delta(K\times T)$ and a finite polyhedral $\BR$-extension, and write
$$
c_i^m(t_i)\coloneqq\AICR_i^m(t_i),
\qquad
R_i^m(t_i)\coloneqq\alpha_i(c_i^m(t_i))
=\ICR_i^m(t_i),
\qquad
R_i^\infty(t_i)\coloneqq\ICR_i^\infty(t_i).
$$
For every $M\geq0$, define
$$
\phi_i^M(t_i)
\coloneqq
\bigl(
R_i^\infty(t_i),c_i^0(t_i),\ldots,c_i^M(t_i)
\bigr),
$$
let $X_i^M\coloneqq\phi_i^M(T_i)$, and define
$$
Q^M
\coloneqq
(\id\times\phi^M)_\#P
\in\Delta(K\times X^M).
$$
We discard types with zero marginal probability. For
$$
x_i=
(R_i^\infty,c_i^0,\ldots,c_i^M)\in X_i^M,
$$
write $R_i^\infty(x_i)=R_i^\infty$ and $c_i^m(x_i)=c_i^m$.

The proof uses two properties of the best-response operator. First, every cell of the tagging map is convex, so conditional averaging within a cell preserves its tag. Second, the best-response operator is monotone with respect to opponents' feasible action sets: enlarging these sets can only enlarge the best-response set. The latter property is established by \citet[Lemma~4.1]{govei2026strategic}.

For $m\geq1$, define the ordinary level-$m$ belief
$$
p_i^m(t_i)
\coloneqq
(\id\times R_{-i}^{m-1})_\#P_i(t_i),
$$
and define the terminal belief
$$
p_i^\infty(t_i)
\coloneqq
(\id\times R_{-i}^\infty)_\#P_i(t_i).
$$

\begin{lemma}[Terminal actions remain best responses]
\label{lem:terminal-actions-best-responses}
For every player $i$ and almost every type $t_i$,
$$
R_i^\infty(t_i)\subseteq\BR_i\bigl(p_i^\infty(t_i)\bigr).
$$
\end{lemma}

\begin{proof}
Under $P_i(t_i)$, the random sets $R_{-i}^{m-1}(t_{-i})$ decrease pointwise to $R_{-i}^\infty(t_{-i})$. Since $\mathcal A_{-i}$ is finite,
$$
\left\|
p_i^m(t_i)-p_i^\infty(t_i)
\right\|_{\mathrm{TV}}
\leq
P_i(t_i)\bigl(
R_{-i}^{m-1}\neq R_{-i}^\infty
\bigr)
\longrightarrow0
$$
for almost every $t_i$.

Fix $a_i\in R_i^\infty(t_i)$. Then $a_i\in R_i^m(t_i)$ and hence
$$
a_i\in\BR_i\bigl(p_i^m(t_i)\bigr)
$$
for every $m\geq1$. By Lemma~\ref{Lem_BRpolytope}, the best-response inclusion region $\mathcal I_i(\{a_i\})$ is closed. Passing to the limit gives
$$
a_i\in\BR_i\bigl(p_i^\infty(t_i)\bigr).
$$
\end{proof}

Let $\widehat c_i^{m,M}(x_i)$ and $\widehat R_i^{m,M}(x_i)$ denote the augmented and ordinary ICR iterates generated by $Q^M$, respectively. Thus
$$
\widehat R_i^{m,M}(x_i)
=
\alpha_i\bigl(\widehat c_i^{m,M}(x_i)\bigr).
$$

\begin{lemma}[Properties of the truncations]
\label{lem:finite-prefix-model}
For every $M\geq0$, the truncation $Q^M$ has the following properties.
\begin{enumerate}[label=(\roman*)]
\item For every player $i$, every positive-probability type $x_i$, and every $0\leq m\leq M$,
$$
\widehat c_i^{m,M}(x_i)=c_i^m(x_i).
$$
Consequently,
$$
\widehat R_i^{m,M}(x_i)
=
\alpha_i(c_i^m(x_i)).
$$
\item For every player $i$, every positive-probability type $x_i$, and every $m\geq0$,
$$
R_i^\infty(x_i)\subseteq\widehat R_i^{m,M}(x_i).
$$
\end{enumerate}
\end{lemma}

\begin{proof}
We first prove (i). The equality holds at level zero because both augmented iterations start at $c_i^*$. Suppose it holds through level $m<M$.

For an original type $t_i$, define
$$
q_i^{m+1}(t_i)
\coloneqq
(\id\times c_{-i}^m)_\#P_i(t_i).
$$
By the definition of the augmented hierarchy,
$$
q_i^{m+1}(t_i)
\in
\beta_i^{-1}\bigl(c_i^{m+1}(t_i)\bigr).
$$

Fix
$$
x_i=(R_i^\infty,c_i^0,\ldots,c_i^M)
$$
and let $\lambda_i^M(dt_i\mid x_i)$ be the conditional distribution of $t_i$ given $\phi_i^M(t_i)=x_i$. By the induction hypothesis and the law of iterated expectations, player $i$'s level-$(m+1)$ belief under $Q^M$ is
$$
\widehat q_i^{m+1,M}(x_i)
=
\int q_i^{m+1}(t_i)\,
\lambda_i^M(dt_i\mid x_i).
$$
Every type in the class $\{\phi_i^M=x_i\}$ has the same coordinate $c_i^{m+1}(x_i)$. Hence
$$
q_i^{m+1}(t_i)
\in
\beta_i^{-1}\bigl(c_i^{m+1}(x_i)\bigr)
$$
for $\lambda_i^M(\cdot\mid x_i)$-almost every $t_i$. The cell $\beta_i^{-1}(c_i^{m+1}(x_i))$ is convex, so the average $\widehat q_i^{m+1,M}(x_i)$ belongs to the same cell. Therefore
$$
\widehat c_i^{m+1,M}(x_i)=c_i^{m+1}(x_i),
$$
which proves (i).

We next prove (ii). Conditional on $x_i$, player $i$'s belief under $Q^M$ about the state and the opponents' declared terminal \ICRt\ sets is
$$
\overline p_i^M(x_i)
\coloneqq
\marg_{k,R_{-i}^\infty}Q^M(\cdot\mid x_i).
$$
Disintegration gives
$$
\overline p_i^M(x_i)
=
\int p_i^\infty(t_i)\,
\lambda_i^M(dt_i\mid x_i).
$$
Every type in this class has terminal \ICRt\ set $R_i^\infty(x_i)$. By Lemma~\ref{lem:terminal-actions-best-responses},
$$
R_i^\infty(x_i)
\subseteq
\BR_i\bigl(p_i^\infty(t_i)\bigr)
$$
for almost every such $t_i$. The best-response inclusion region $\mathcal I_i(R_i^\infty(x_i))$ is convex by Lemma~\ref{Lem_BRpolytope}. Hence
\begin{equation}
R_i^\infty(x_i)
\subseteq
\BR_i\bigl(\overline p_i^M(x_i)\bigr).
\label{eq:pooled-terminal-obedience}
\end{equation}

We prove (ii) by induction on $m$. The claim holds at level zero because $R_i^\infty(x_i)\subseteq A_i$. Suppose that
$$
R_j^\infty(x_j)\subseteq\widehat R_j^{m,M}(x_j)
$$
for every player $j$ and every positive-probability type $x_j$. Conditional on $x_i$, consider the joint distribution under $Q^M$ of
$$
\bigl(
k,R_{-i}^\infty(x_{-i}),\widehat R_{-i}^{m,M}(x_{-i})
\bigr).
$$
The induction hypothesis gives
$$
R_{-i}^\infty(x_{-i})
\subseteq
\widehat R_{-i}^{m,M}(x_{-i})
$$
almost surely. By \eqref{eq:pooled-terminal-obedience} and the monotonicity of $\BR_i$ established by \citet[Lemma~4.1]{govei2026strategic},
$$
R_i^\infty(x_i)
\subseteq
\BR_i\left(
\marg_{k,\widehat R_{-i}^{m,M}}
Q^M(\cdot\mid x_i)
\right)
=
\widehat R_i^{m+1,M}(x_i).
$$
\end{proof}

\begin{proof}[Proof of Theorem~\ref{thm:finite-type-approximation}]
Define
$$
G_M
\coloneqq
\left\{
t\in T:
R_i^M(t_i)=R_i^\infty(t_i)
\text{ for every }i
\right\}.
$$
Consider a type $x_i=\phi_i^M(t_i)$ such that $R_i^M(t_i)=R_i^\infty(t_i)$. By Lemma~\ref{lem:finite-prefix-model}(i),
$$
\widehat R_i^{M,M}(x_i)
=
R_i^M(t_i)
=
R_i^\infty(x_i).
$$
Lemma~\ref{lem:finite-prefix-model}(ii) and the decreasing monotonicity of the ICR iteration imply that, for every $m\geq M$,
$$
R_i^\infty(x_i)
\subseteq
\widehat R_i^{m,M}(x_i)
\subseteq
\widehat R_i^{M,M}(x_i)
=
R_i^\infty(x_i).
$$
Hence
\begin{equation}
\widehat R_i^{\infty,M}(x_i)=R_i^\infty(x_i).
\label{eq:terminal-on-stabilized-prefix}
\end{equation}

Couple $P$ and $Q^M$ by drawing $(k,t)$ according to $P$ and setting $x=\phi^M(t)$. Under this coupling, the distribution of $(k,R^\infty(t))$ is $\nu_P$. On $G_M$, equation~\eqref{eq:terminal-on-stabilized-prefix} holds for every player, so the terminal \ICRt\ profile generated under $Q^M$ equals $R^\infty(t)$. Therefore
$$
\|\nu_{Q^M}-\nu_P\|_{\mathrm{TV}}
\leq
P(G_M^\complement).
$$

For every player $i$ and every type $t_i$, the decreasing sequence $(R_i^m(t_i))_m$ consists of subsets of the finite set $A_i$ and therefore eventually equals its intersection $R_i^\infty(t_i)$. Hence
$$
\boldsymbol 1_{\{R_i^M\neq R_i^\infty\}}
\longrightarrow0
$$
pointwise. Since the player set is finite, dominated convergence gives
$$
P(G_M^\complement)\longrightarrow0.
$$
This proves the approximation claim.

Each $Q^M$ has finite type spaces. Its induced augmented-hierarchy distribution therefore has finite support and satisfies condition~(2) of Theorem~\ref{RPGeneral}. It remains to verify eventual constancy. For every positive-probability type, the ordinary ICR hierarchy eventually stabilizes. Since the type spaces are finite, there is a common round $L_M$ after which $\widehat R_i^{m,M}(x_i)$ is constant in $m$ for every player and every positive-probability type. For $m\geq L_M$, the belief of $x_i$ about
$$
\bigl(k,\widehat R_{-i}^{m,M}(x_{-i})\bigr)
$$
is independent of $m$. The tagging map $\tau_i$ therefore assigns the same tag at every subsequent round. Thus every augmented hierarchy generated by $Q^M$ is eventually constant.
\end{proof}

\subsection{A Countable Outcome-Dense Family of Types}
\label{app:countable-dense}

We prove Theorem~\ref{thm:countable-dense-types}. The proof works separately at each truncation level. Since the map from a common prior to its terminal \ICRt\,distribution need not be continuous, density in the simplex of priors is insufficient. We instead approximate the graph of this map.

Fix a finite polyhedral $\BR$-extension, which exists by Lemma~\ref{lem:exist-extension}. For every $n\in\naturals$, let
$$
Y_i^n
\coloneqq
\mathcal A_i\times\mathcal C_i^{n+1},
$$
and let $\mathcal Q_n\subseteq\Delta(K\times Y^n)$ be the set of common priors that arise as $n$-truncations. For $Q\in\mathcal Q_n$, denote by $\nu_Q\in\Delta(K\times\mathcal A)$ its terminal ICR distribution, and define
$$
\Gamma_n
\coloneqq
\left\{
(Q,\nu_Q):Q\in\mathcal Q_n
\right\}
\subseteq
\Delta(K\times Y^n)\times\Delta(K\times\mathcal A).
$$

\begin{lemma}[Countable approximation at each truncation level]
\label{lem:countable-truncation-stratum}
For every $n$, there exists a countable family $(Q^{m,n})_{m\in\naturals}$ in $\mathcal Q_n$ such that
$$
\left\{
(Q^{m,n},\nu_{Q^{m,n}}):m\in\naturals
\right\}
$$
is dense in $\Gamma_n$.
\end{lemma}

\begin{proof}
The product
$$
\Delta(K\times Y^n)\times\Delta(K\times\mathcal A)
$$
is finite dimensional and hence second countable. Its subspace $\Gamma_n$ is therefore second countable and hence separable.
\end{proof}

For each $n$, fix a family as in Lemma~\ref{lem:countable-truncation-stratum}, and define
$$
\mathbb Q
\coloneqq
\left\{
Q^{m,n}:m,n\in\naturals
\right\}.
$$
For each player $i$, let $\mathbb T_i$ be the disjoint union of the finite type spaces used by the common priors in $\mathbb Q$. Each selected common prior embeds in these common countable type spaces by placing its distribution on its own component; its conditional beliefs on that component are unchanged. Thus $\mathbb Q$ and every $\mathbb T_i$ are countable. Alternatively, one may embed all selected common priors in the universal Harsanyi type space and take the union of their images.

\begin{proof}[Proof of Theorem~\ref{thm:countable-dense-types}]
Let $P$ be a common prior, and let $\widehat Q^n$ be its $n$-truncation. By Theorem~\ref{thm:finite-type-approximation},
$$
\|\nu_{\widehat Q^n}-\nu_P\|_{\mathrm{TV}}
\longrightarrow0.
$$
Since $(\widehat Q^n,\nu_{\widehat Q^n})\in\Gamma_n$, Lemma~\ref{lem:countable-truncation-stratum} allows us to choose $m(n)$, for every $n\geq1$, such that
$$
\|Q^{m(n),n}-\widehat Q^n\|_{\mathrm{TV}}
+
\|\nu_{Q^{m(n),n}}-\nu_{\widehat Q^n}\|_{\mathrm{TV}}
<
\frac1n.
$$
It follows that
$$
\|\nu_{Q^{m(n),n}}-\nu_P\|_{\mathrm{TV}}
\leq
\|\nu_{Q^{m(n),n}}-\nu_{\widehat Q^n}\|_{\mathrm{TV}}
+
\|\nu_{\widehat Q^n}-\nu_P\|_{\mathrm{TV}}
\longrightarrow0.
$$
The sequence $(Q^{m(n),n})_n$ belongs to $\mathbb Q$ and uses only types in $(\mathbb T_i)_i$.
\end{proof}

The families $\mathbb Q$ and $(\mathbb T_i)_i$ depend on the fixed payoff structure. The theorem concerns density of terminal \ICRt\,outcome distributions in that payoff structure. It does not assert density of $(\mathbb T_i)_i$ in the strategic topology on universal types, which compares behavior across payoff structures.

\end{document}